\documentclass[pra,aps,superscriptaddress,nofootinbib,twocolumn]{revtex4-1}

\usepackage{mathrsfs}
\usepackage{amsfonts}
\usepackage{amssymb}
\usepackage{amsmath}
\usepackage{graphicx}
\usepackage[usenames,dvipsnames]{color}
\usepackage[colorlinks=true,citecolor=blue,linkcolor=magenta]{hyperref}
\usepackage{ulem}
\usepackage{lmodern}
\usepackage{amsthm}

\newcommand{\figpath}{.}

\newcommand{\Tr}{\mathrm{Tr}}
\newcommand{\norm}[1]{\Vert #1 \Vert}
\newcommand{\abs}[1]{\vert #1 \vert}

\newcommand{\ket}[1]{\vert{ #1 }\rangle}

\newcommand{\ketbra}[2]{\vert #1 \rangle \langle #2 \vert}

\newcommand{\rket}[1]{\vert{ #1 }\rangle\rangle}
\newcommand{\rbra}[1]{\langle\langle{ #1 }\vert}
\newcommand{\rbraket}[2]{\langle\langle #1 \vert #2 \rangle\rangle}
\newcommand{\rketbra}[2]{\vert #1 \rangle\rangle \langle\langle #2 \vert}

\newtheorem{theorem}{Theorem}
\newtheorem{lemma}{Lemma}

\begin{document}

\title{Self-consistent tomography of temporally correlated errors}

\author{Mingxia Huo}

\affiliation{Beijing Computational Science Research Center, Beijing 100193, China}

\affiliation{Department of Physics and Beijing Key Laboratory for Magneto-Photoelectrical Composite and Interface Science, School of Mathematics and Physics, University of Science and Technology Beijing, Beijing 100083, China}

\author{Ying Li}

\affiliation{Graduate School of China Academy of Engineering Physics, Beijing 100193, China}

\begin{abstract}
The error model of a quantum computer is essential for optimizing quantum algorithms to minimize the impact of errors using quantum error correction or error mitigation. Noise with temporal correlations, e.g.~low-frequency noise and context-dependent noise, is common in quantum computation devices and sometimes even significant. However, conventional tomography methods have not been developed for obtaining an error model describing temporal correlations. In this paper, we propose self-consistent tomography protocols to obtain a model of temporally correlated errors, and we demonstrate that our protocols are efficient for low-frequency noise and context-dependent noise. 
\end{abstract}

\maketitle

\section{Introduction}

How to correct errors is one of the most critical issues in practical quantum computation. In the theory of quantum fault tolerance based on quantum error correction (QEC), an arbitrarily high-fidelity quantum computation can be achieved, providing sub-threshold error rates and sufficient qubits~\cite{Nielsen2010}. Recently, error rates within or close to the sub-threshold regime have been demonstrated in various platforms~\cite{Barends2014, Rong2015, Lucas, Wineland, BlumeKohout2017}. These error rates are measured using either randomized benchmarking (RB)~\cite{Emerson2005, Knill2008, Magesan2011, Wallman2014, Fogarty2015, Ball2016, Mavadia2018} or quantum process tomography (QPT)~\cite{Poyatos1997, Chuang1997}. RB only estimates an average effect of the noise, and QPT can provide a model of error channels. Rigorously speaking, whether or not a quantum system is in the sub-threshold regime is not only determined by the error rate but also the detailed error model~\cite{Wang2011, Kueng2016}, including correlations between errors~\cite{Aharonov1999}. Therefore, an error model describing correlated errors is important for verifying sub-threshold quantum devices. We can also optimize QEC protocols by exploring these correlations~\cite{Wang2011, Huo2017}, which is crucial for the early-stage demonstration of small-scale quantum fault tolerance. Given the limited number of qubits and error rate close to the threshold, we need to carefully choose the protocol to observe any advantage of QEC~\cite{Chiaverini2004, Schindler2011, Nigg2014, Taminiau2014, Corcoles2015, Riste2015, Muller2016, Linke2017, Bermudez2017}. 

We may still need many years to realise a fault-tolerant quantum computer~\cite{Fowler2012, Joe2017}, however noisy intermediate-scale quantum (NISQ) computers are likely to be developed in the near future~\cite{Preskill2018, Boixo2016, Neill2018}. Quantum error mitigation is an alternative approach to high-fidelity quantum computation~\cite{Li2017, Temme2017, Endo2017, Kandala2018}, which does not require encoding, therefore, is more promising than QEC on NISQ devices. In quantum error mitigation using the error extrapolation, we can increase errors to learn their effect on the observable representing the computation result. Once we know how the observable changes with the level of errors, we can make an extrapolated estimate of the zero-error computation result. This extrapolation can be implemented directly on the final result or each gate using the quasi-probability decomposition formalism. The effect of errors depends on the error model. Therefore, we need to increase errors according to the model of original errors in the system, i.e.~at first we need a proper error model of original errors. The error model can be obtained using gate set tomography (GST)~\cite{Merkel2013, BlumeKohout2013, Stark2014, Greenbaum2015, BlumeKohout2017, Sugiyama2018}, a self-consistent QPT protocol. With the self-consistent error model, the effect of errors on the computation result can be eliminated, under the condition that errors are not correlated~\cite{Endo2017}. However, correlations are common in quantum systems~\cite{Hooge1981, Paik2011, Sank2012}, e.g.~the slow drift of laser frequency can cause time-dependent gate fidelity in ion trap systems~\cite{Rutman1978, Wineland1998, SchmidtKaler2003, Benhelm2008, Ballance2016}, which limits the fidelity of quantum computation on NISQ devices. Neither RB nor conventional QPT can provide an error model describing temporal correlations~\cite{Wallman2014, Fogarty2015, Ball2016, BlumeKohout2017, Mavadia2018}. 

In this paper, we propose self-consistent tomography protocols to obtain the model of temporally correlated errors without using any additional operations accessing the environment. Temporal correlations are caused by the correlations between the system and the environment. Without a set of informationally-complete state preparation and measurement operations, we cannot implement conventional QPT on the environment. We find that quantum gates are fully characterized by a set of linear operators acting on a subspace of Hermitian matrices, which can be measured in the experiment only using themselves even without information completeness. However, these operators may not be complete completely positive (CP) maps as in conventional QPT. We term our method as linear operator tomography (LOT), which can be used to reconstruct an operator representation of quantum gates without using additional operations on the environment. A tremendous amount of experimental data may be required to obtain the exact model of temporally correlated errors. For the practical implementation, we aim at an approximate error model, and we find that efficient approximations for low-frequency time-dependent noise and context-dependent noise exist~\cite{Rudinger2018, Veitia2018}. Practical protocols are proposed and demonstrated numerically. Error rates estimated using RB and GST may exhibit significant difference due to temporal correlations~\cite{BlumeKohout2017}. In numerical simulations, we show that RB and LOT results coincide with each other even in the presence of temporal correlations. 

The paper is organized as follows. In Sec.~\ref{sec:model}, we first introduce a general model of a quantum computer including the environment, wherein temporal correlations are caused by the environment. In Sec.~\ref{sec:exact}, we show that in principle by only using operations for operating the system, we can reconstruct a self-consistent model of both the system and the environment. The exact LOT protocol is introduced in Sec.~\ref{sec:exact}. In Sec.~\ref{sec:truncation}, we present the condition for performing a truncation on the system-environment state space. In Sec.~\ref{sec:app_model}, we discuss low-frequency time-dependent noise and context-dependent noise. We show that a low-dimensional state space can characterize these two types of temporally correlations. In Sec.~\ref{sec:appQT}, we give two approximate tomography protocols for the practical implementation. In Sec.~\ref{sec:simulation}, we demonstrate the protocol in numerical simulations. 

\section{The model of a quantum computer}
\label{sec:model}

\begin{figure}[tbp]
\centering
\includegraphics[width=1\linewidth]{\figpath 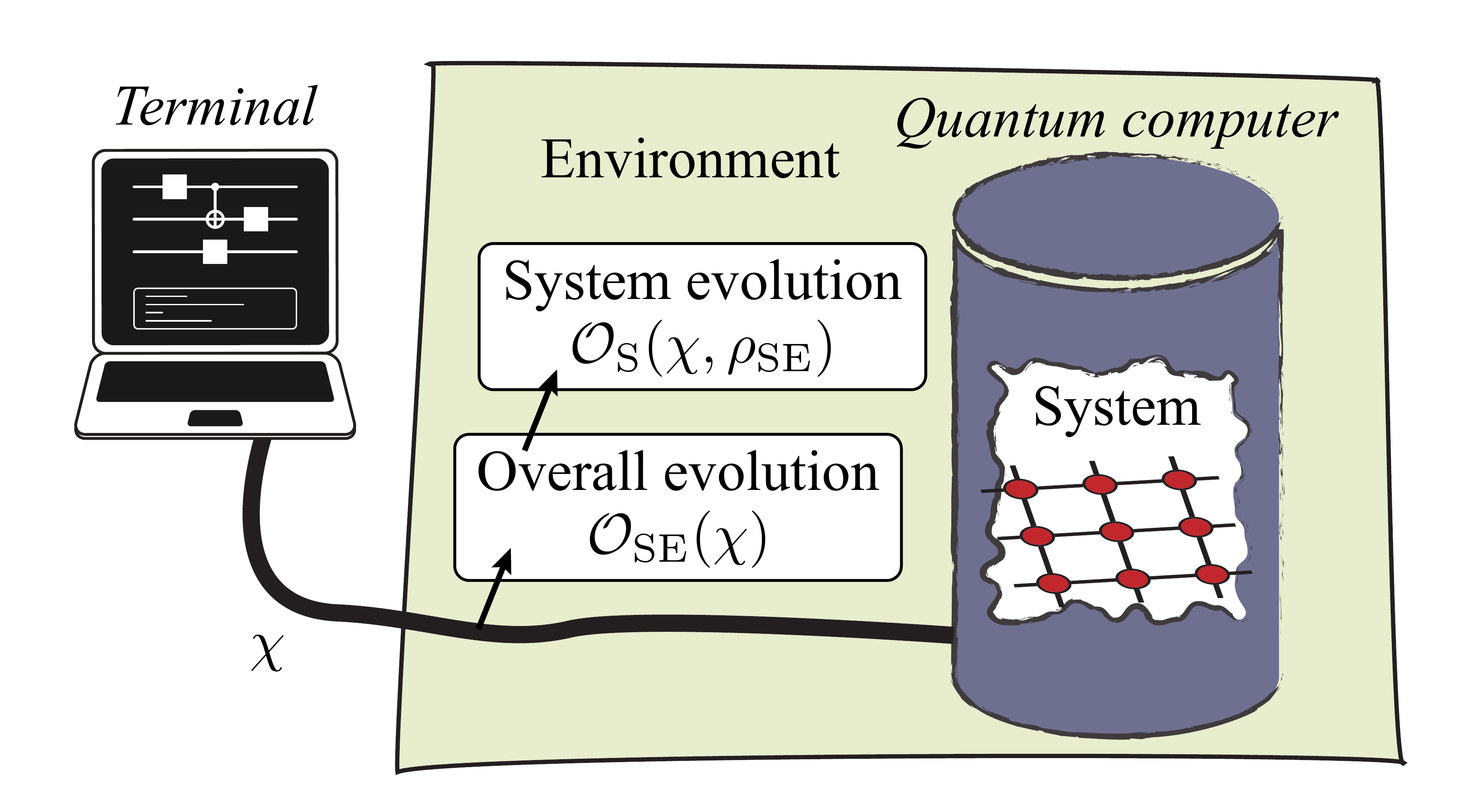}
\caption{
Quantum computer controlled by the state of a terminal. The state of the terminal $\chi$ results in the evolution $\mathcal{O}_{\rm SE}(\chi)$ of the system and environment. The evolution of the system $\mathcal{O}_{\rm S}(\chi, \rho_{\rm SE})$ depends on both the terminal state $\chi$ and the state of the system and environment at the beginning of the evolution $\rho_{\rm SE}$. 
}
\label{fig:model}
\end{figure}

For illustrating how to describe errors with temporal correlations, we start with an example in ion trap systems. If the quantum gate is driven by the laser field, usually the gate fidelity depends on the laser frequency $\lambda$~\cite{Rutman1978, Wineland1998, SchmidtKaler2003, Benhelm2008, Ballance2016}. The laser frequency drifts with time, and usually we are not aware of its instant value. We use $\lambda$ to denote a time-dependent random variable, such as the laser frequency. The stochastic process of $\lambda$ is characterized by the probability distribution $\bar{p}(f)$, i.e.~the value of $\lambda$ at time $t$ is $f(t)$, and $\bar{p}(f)$ is the probability density in the space of functions $f(t)$. We use the superoperator $\mathcal{O}_{\rm S}(\lambda)$ to denote the gate operation given by the laser frequency $\lambda$. With the initial state $\rho_{\rm S}$, the output state after two gates $\mathcal{O}_{\rm S}$ and $\mathcal{O}_{\rm S}'$ at $t$ and $t'$, respectively, reads $\rho_{\rm S}^{(2)} = \int df \bar{p}(f) \mathcal{O}_{\rm S}'(f(t')) \mathcal{O}_{\rm S}(f(t))(\rho_{\rm S})$. We can find that the state $\rho_{\rm S}^{(2)}$ cannot be factorized as two independent operations on the initial state. Therefore, conventional QPT cannot be applied~\cite{Poyatos1997, Chuang1997, Merkel2013, BlumeKohout2013, Stark2014, Greenbaum2015, BlumeKohout2017, Sugiyama2018}. 

For simplification, we assume that $\lambda$ changes slowly with time, and the typical time that $\lambda$ changes is much longer than the time scale of a quantum circuit, i.e.~the time from the state preparation to the measurement. In this case, we neglect the change of $\lambda$ within each run of the quantum circuit, i.e.~$f(t') = f(t)$ for two gates in the same run. We also assume that the distribution is stationary, then the distribution of the instant value, i.e.~$p(\lambda) = \int df \bar{p}(f) \delta(\lambda-f(t))$, is time-independent. With the distribution of the instant value, we can rewrite the output state in the form $\rho_{\rm S}^{(2)} = \int d\lambda p(\lambda) \mathcal{O}_{\rm S}'(\lambda) \mathcal{O}_{\rm S}(\lambda)(\rho_{\rm S})$. 

We can factorize multi-gate superoperators by introducing the state space of the laser frequency. We use $\ket{\lambda}_{\rm E}$ to denote the state corresponding to the laser frequency $\lambda$. The state space of $\ket{\lambda}_{\rm E}$ can be virtual, i.e.~it is not necessary that $\ket{\lambda}_{\rm E}$ corresponds to a pure state in a physical Hilbert space. The initial state of qubits and the laser frequency can be expressed as $\rho_{\rm SE} = \int d\lambda p(\lambda) \rho_{\rm S} \otimes \ketbra{\lambda}{\lambda}_{\rm E}$. Then the superoperator of a laser-frequency-dependent gate can be expressed as $\mathcal{O}_{\rm SE} = \int d\lambda \mathcal{O}_{\rm S}(\lambda) \otimes [\ketbra{\lambda}{\lambda}_{\rm E}]$, where $[U]$ denotes a superoperator $[U](\rho) = U \rho U^\dag$. After two gates, the state of the system and the laser frequency reads $\rho_{\rm SE}^{(2)} = \mathcal{O}_{\rm SE}' \mathcal{O}_{\rm SE} (\rho_{\rm SE})$, which is in the factorized form. One can find that $\rho_{\rm S}^{(2)} = \Tr_{\rm E} (\rho_{\rm SE}^{(2)})$. 

We note that multi-gate superoperators can be factorized following a similar procedure for noise with any spectrum, i.e.~the change of the variable $\lambda$ in the time scale of a quantum circuit can be nonnegligible or even significant, as we will show in Sec.~\ref{sec:low-frequency}. It is straightforward to generalise the approach to the case of multiple random variables, e.g.~the gate fidelity depends on a set of drifting laser parameters, and the case that the initial state of the system also depends on random variables. 

By introducing the environment, which is the frequency space in the example, we can describe temporally correlated errors. Such a formalism has been used in an ion-trap tomography experiment~\cite{BlumeKohout2017}, in which a classical bit is introduced to represent the environment memory. With one classical bit and one qubit, the tomography is implemented on an eight-dimensional state space. 

Now, we introduce a general model of quantum computer. For a quantum computer with $n$ qubits, we call the $2^n$-dimensional Hilbert space of qubits the system. Degrees of freedom coupled to the system form the environment, including but not limited to all random variables determining the evolution of the system. Quantum computation is realised by a sequence of quantum gates. The gate sequence is stored in a terminal, e.g.~a classical computer, and the evolution of the system-environment (SE) is controlled by the terminal as shown in Fig.~\ref{fig:model}. We use $\chi$ to denote the state of the terminal indicating ``Implement the gate $\chi$'' and superoperator $\mathcal{O}_{\rm SE}(\chi)$ to denote the corresponding evolution of SE. Here $\chi$ is a deterministic parameter rather than a random variable. By setting $\chi$ according to the gate sequence, we realise the quantum computation. We assume that the Born-Markov approximation can be applied to the terminal and SE, and operations on SE are Markovian and factorized. For the gate sequence $\chi_1,\chi_2,\ldots,\chi_N$, the overall evolution of SE is $\mathcal{O}_{\rm SE}(\chi_N)\cdots\mathcal{O}_{\rm SE}(\chi_2)\mathcal{O}_{\rm SE}(\chi_1)$. 

In this model, the time dependence for operations on the system are not expressed explicitly. Operations on SE are time-independent. However, corresponding system operations are stochastic and depend on the environment state. When the environment state evolves, which is driven by SE operations, system operations evolves accordingly. In Sec.~\ref{sec:low-frequency}, we will give an example that SE operations drive the stochastic process of the environment. In this way, we can describe errors with general temporal correlations, not only correlations caused by classical random variables such as laser frequencies, but also correlations caused by the coupling to a quantum system in the environment. 

In general, the evolution of the system $\mathcal{O}_{\rm S}(\chi, \rho_{\rm SE})$ depends on not only $\chi$ but also the state of SE at the beginning of the evolution $\rho_{\rm SE}$. If the system and environment are correlated in $\rho_{\rm SE}$, the system evolution may not even be CP~\cite{Pechukas1994}. If the system evolution $\mathcal{O}_{\rm S}(\chi, \rho_{\rm SE}) = \mathcal{O}_{\rm S}(\chi)$ does not depend on $\rho_{\rm SE}$, the overall system evolution of a gate sequence is $\mathcal{O}_{\rm S}(\chi_N)\cdots\mathcal{O}_{\rm S}(\chi_2)\mathcal{O}_{\rm S}(\chi_1)$. In this case, conventional QPT can be applied, we can obtain $\mathcal{O}_{\rm S}(\chi)$ up to a similarity transformation using GST~\cite{Merkel2013, BlumeKohout2013, Stark2014, Greenbaum2015, BlumeKohout2017, Sugiyama2018}, and the computation error can be mitigated as proposed in Ref.~\cite{Endo2017}. By introducing the environment, non-Markovian processes can be reconstructed using quantum tomography~\cite{Pollock2018}. 

From now on, we focus on states and operations of SE and neglect the subscript `SE'. All states, operations and observables without a subscript (`SE', `S' or `E') correspond to SE; and subscripts `S' and `E' are used to denote the system and the environment, respectively. 

We would like to remark that LOT protocols proposed in this paper cannot reconstruct the complete CP maps acting on SE as in conventional QPT protocols. In LOT, we only use the operations designed to operate the system, i.e.~quantum gates for the computation, which actually act on SE because of imperfections. Given the limited accessibility to the environment, it is unrealistic to implement informationally-complete state preparation and measurement for the full tomography of SE. 

\subsection*{State, measurement, operations and Pauli transfer matrix representation}

A quantum computer is characterized by a set of linear operators: the initial state $\rho_{\rm in}$ which is a normalized positive Hermitian operator, the measurement (i.e.~measured observable) $Q_{\rm out}$ which is also a Hermitian operator, and a set of elementary computation operations $\{ \mathcal{O}(\chi) \}$ which are CP maps. We remark that, $\rho_{\rm in}$, $Q_{\rm out}$ and $\{ \mathcal{O}(\chi) \}$ describe the actual quantum computer (including both the system and environment) rather than an ideal quantum computer, and they are all unknown therefore need to be investigated in tomography. The quantum computation is realised by a sequence of elementary operations on the initial state. The set of operation sequences $O = \{ \mathcal{O}(\chi_N) \cdots \mathcal{O}(\chi_2) \mathcal{O}(\chi_1) \}$ includes all operations generated by elementary operations $\{ \mathcal{O}(\chi) \}$. 

We focus on the case that the quantum computer only provides one option of the initial state $\rho_{\rm in}$ and one option of the observable to be measured $Q_{\rm out}$. It is straightforward to generalize our results to the case that multiple options of initial states and observables are available. 

In this paper, we use Pauli transfer matrix representation~\cite{Merkel2013, BlumeKohout2013, Stark2014, Greenbaum2015, BlumeKohout2017, Sugiyama2018}: $\rket{\rho}$ is a column vector with elements $\rket{\rho}_\sigma = \Tr(\sigma\rho)$; $\rbra{Q}$ is a row vector with elements $\rbra{Q}_\sigma = d_{\rm H}^{-1} \Tr(Q\sigma)$; then an quantum operation $\mathcal{O}$ can be expressed as a square matrix with elements $\mathcal{O}_{\sigma,\tau} = d_{\rm H}^{-1} \Tr[\sigma\mathcal{O}(\tau)]$. Here, $\sigma$ and $\tau$ are Pauli operators or generalized Pauli operators, i.e.~Hermitian operators satisfying $\Tr(\sigma\tau) = d_{\rm H} \delta_{\sigma,\tau}$, and $d_{\rm H}$ is the dimension of the Hilbert space of SE. All these vectors and matrices are real and $d_{\rm H}^2$-dimensional. For a state $\rho$ and an observable $Q$, $\rbraket{Q}{\rho} = \Tr(Q\rho)$ is the mean of the observable $Q$ in the state $\rho$. For an operation $\mathcal{O}$, $\rket{\mathcal{O}(\rho)} = \mathcal{O}\rket{\rho}$ is the vector corresponding to the state $\mathcal{O}(\rho)$. Therefore, the mean of the observable $Q$ in the state $\rho$ after a sequence of quantum operations reads $\Tr[Q\mathcal{O}_N\cdots\mathcal{O}_2\mathcal{O}_1(\rho)] = \rbra{Q}\mathcal{O}_N\cdots\mathcal{O}_2\mathcal{O}_1\rket{\rho}$. 

\section{Self-consistent tomography without information completeness}
\label{sec:exact}

Information completeness is required by conventional QPT protocols. If we can prepare a complete set of states $\{ \rket{\rho_i} = \mathcal{O}_i \rket{\rho_{\rm in}} \}$ and measure a complete set of observables $\{ \rbra{Q_k} = \rbra{Q_{\rm out}} \mathcal{O}_k' \}$, i.e.~$d_{\rm H}^2$ linearly independent vectors in each set, we can implement QPT on SE to reconstruct the CP maps of quantum gates. Here, the CP maps act on SE. However, information completeness for SE is unrealistic. 

In this section, we demonstrate that it is possible to exactly characterize a set of quantum gates using tomography without information completeness. In the LOT formalism, we obtain a set of operators acting on a subspace of Hermitian matrices to represent quantum gates, which may not be complete CP maps without information completeness. Such an operator representation is adequate in the sense that given the initial state, an arbitrary operation sequence and the observable to be measured, the average value of the observable can be computed using these operators. 

Because information completeness is not required, we can use LOT to characterize the quantum computer even in the presence of temporal correlations. In this section, the feasibility is not our concern. In the following sections, we will discuss how to adapt the protocol for the purpose of practical implementation. 

\subsection*{Linear operator tomography}

With only computation operations, which are designed to operate the system but act on SE because of imperfections, usually we do not have complete state and observable sets, therefore, we cannot access the entire space of Hermitian matrices. 

Given an initial state $\rho_{\rm in}$, an observable $Q_{\rm out}$ and a set of elementary operations $\{ \mathcal{O}(\chi) \}$, we consider three subspaces of Hermitian matrices as follows. The subspace $V_{\rm in} = {\rm span} ( \{ \mathcal{O} \rket{\rho_{\rm in}} ~\vert~ \mathcal{O} \in O \} )$ is the span of all states that can be prepared in the quantum computer, and the subspace $V_{\rm out} = {\rm span} ( \{ \rbra{Q_{\rm out}} \mathcal{O} ~\vert~ \mathcal{O} \in O \} )$ is the span of all observables that can be effectively measured. Note that $O$ is the set of all operations generated by elementary operations. We use $P_{\rm in}$ and $P_{\rm out}$ to denote the orthogonal projections on $V_{\rm in}$ and $V_{\rm out}$, respectively. The third subspace is $V = {\rm span} ( \{ P_{\rm out} \mathcal{O} \rket{\rho_{\rm in}} ~\vert~ \mathcal{O} \in O \} )$, and we use $P$ to denote the orthogonal projection on $V$. 

The subspace $V$ is the space of Hermitian matrices that the finite set of operations can access to. If $V$ is the entire Hermitian-matrix space with the dimension $d_{\rm H}^2$, states and observables are complete, then LOT is the same as GST. In general, the completeness is not required in LOT. 

Our first result is that in order to fully characterize the quantum computer, we only need to reconstruct $P \rket{\rho_{\rm in}}$, $\rbra{Q_{\rm out}} P$ and $\{ P \mathcal{O}(\chi) P \}$ in the tomography. The reason is that, for an arbitrary sequence of operations in $O$, we have
\begin{eqnarray}
&& \rbra{Q_{\rm out}} \mathcal{O}_N \cdots \mathcal{O}_2 \mathcal{O}_1 \rket{\rho_{\rm in}} \notag \\
&=& \rbra{Q_{\rm out}} P \mathcal{O}_N P \cdots P \mathcal{O}_2 P \mathcal{O}_1 P \rket{\rho_{\rm in}}.
\end{eqnarray}
See Appendix~\ref{app:exact} for the proof. We would like to remark that the conclusion is the same for the subspace ${\rm span} ( \{ \rbra{Q_{\rm out}} \mathcal{O} P_{\rm in} ~\vert~ \mathcal{O} \in O \} )$. 

With this result, we can perform the tomography in a similar way to GST. We note that the protocol presented in this section is for the purpose of illustrating the self-consistent formalism rather than the practical implementation, and we discuss the practical implementation later. We need to assume that the dimension of the subspace $V$ is finite and known, see discussions at the end of this section. The dimension of $V$ is $d_V = \Tr(P)$. The exact LOT protocol is as follows: 
\begin{itemize}
\item[$\bullet$] Choose a set of states $\{ \rket{\rho_i} = \mathcal{O}_i \rket{\rho_{\rm in}} \}$ and a set of observables $\{ \rbra{Q_k} = \rbra{Q_{\rm out}} \mathcal{O}_k' \}$. Here, $\mathcal{O}_i, \mathcal{O}_k' \in O$, $i,k = 1,\cdots,d$ and we take $d = d_V$. We always take $\rket{\rho_1} = \rket{\rho_{\rm in}}$ and $\rbra{Q_1} = \rbra{Q_{\rm out}}$. 
\item[] These states and observables must satisfy the condition that $\{ P \rket{\rho_i} \}$ and $\{ \rbra{Q_k} P \}$ are both linearly independent. According to the definition of the subspace $V$, states and observables satisfying the condition always exist and can be realised in the quantum computer using the combination of elementary operations. 
\item[$\bullet$] Obtain matrices $g = M_{\rm out} M_{\rm in}$ and $\widetilde{\mathcal{O}}(\chi) = M_{\rm out} \mathcal{O}(\chi) M_{\rm in}$ for each $\chi$ in the experiment. Here, $M_{\rm in} = [ \; \rket{\rho_1} \;\cdots\; \rket{\rho_d} \; ]$ is the matrix with $\{ \rket{\rho_i} \}$ as columns, and $M_{\rm out} = [ \; \rbra{Q_1}^{\rm T} \;\cdots\; \rbra{Q_d}^{\rm T} \; ]^{\rm T}$ is the matrix with $\{ \rbra{Q_k} \}$ as rows. Each matrix element can be measured in the experiment. The element $g_{k,i} = \rbraket{Q_k}{\rho_i}$ is the mean of $Q_k$ in the state $\rho_i$. The element $\widetilde{\mathcal{O}}_{k,i}(\chi) =  \rbra{Q_k} \mathcal{O}(\chi) \rket{\rho_i}$ is the mean of $Q_k$ in the state $\rho_i$ after the operation $\mathcal{O}(\chi)$. 
\item[] When $\{ P \rket{\rho_i} \}$ and $\{ \rbra{Q_k} P \}$ linearly independent, $g$ is invertible. 
\end{itemize}

Data $g$ and $\{ \widetilde{\mathcal{O}}(\chi) \}$ exactly characterize the quantum computer. We have
\begin{eqnarray}
M_{\rm out} \mathcal{O}_N \cdots \mathcal{O}_2 \mathcal{O}_1 M_{\rm in} = \widetilde{\mathcal{O}}_N g^{-1} \cdots \widetilde{\mathcal{O}}_2 g^{-1} \widetilde{\mathcal{O}}_1
\label{eq:main_GST}
\end{eqnarray}
for an arbitrary sequence of operations in $O$. See Appendix~\ref{app:exact} for the proof. 

Given $g$ and $\{ \widetilde{\mathcal{O}}(\chi) \}$, we can obtain an exact error model of the quantum computer.
\begin{itemize}
\item[$\bullet$] Choose a $d$-dimensional invertible real matrix $\widehat{M}_{\rm in}$, and compute $\widehat{M}_{\rm out} = g\widehat{M}_{\rm in}^{-1}$.
\item[$\bullet$] Take $\rket{\widehat{\rho}_{\rm in}} = \widehat{M}_{\rm in; \bullet,1}$ and $\rbra{\widehat{Q}_{\rm out}} = \widehat{M}_{\rm out; 1,\bullet}$, and compute $\widehat{\mathcal{O}}(\chi) = \widehat{M}_{\rm out}^{-1} \widetilde{\mathcal{O}}(\chi) \widehat{M}_{\rm in}^{-1}$ for each $\chi$.
\end{itemize}
Here, $M_{\bullet,i}$ and $M_{k,\bullet}$ denote the $i^{\rm th}$ column and the $k^{\rm th}$ row of the matrix $M$, respectively. The error model of the quantum computer is formed by $\rket{\widehat{\rho}_{\rm in}}$, $\rbra{\widehat{Q}_{\rm out}}$ and $\{ \widehat{\mathcal{O}}(\chi) \}$, which correspond to $\rket{\rho}$, $\rbra{Q}$ and $\{ \mathcal{O}(\chi) \}$, respectively. 

According to Eq.~(\ref{eq:main_GST}), we have
\begin{eqnarray}
&&\rbra{\widehat{Q}_{\rm out}} \widehat{\mathcal{O}}(\chi_N) \cdots \widehat{\mathcal{O}}(\chi_1) \rket{\widehat{\rho}_{\rm in}} \notag \\
&=& \rbra{Q_{\rm out}} \mathcal{O}(\chi_N) \cdots \mathcal{O}(\chi_1) \rket{\rho_{\rm in}}.
\end{eqnarray}
The first line is the computation result according to the error model (the tomography result), and the second line is the experimental result of the actual quantum computer, which are equal. In this sense the error model is exact. The exactness of the error model does not rely on how to choose the matrix $\widehat{M}_{\rm in}$. If we choose a different matrix $\widehat{M}_{\rm in}'$, then we can obtain another error model $\rket{\widehat{\rho}_{\rm in}'} = S\rket{\widehat{\rho}_{\rm in}}$, $\rbra{\widehat{Q}_{\rm out}'} = \rbra{\widehat{Q}_{\rm out}}S^{-1}$ and $\{ \widehat{\mathcal{O}}'(\chi) = S\widehat{\mathcal{O}}(\chi)S^{-1} \}$, where $S = \widehat{M}_{\rm in}'\widehat{M}_{\rm in}^{-1}$. Both error models can exactly characterize the quantum computer, because the difference between two error models is only a similarity transformation~\cite{Merkel2013, BlumeKohout2013, Stark2014, Greenbaum2015, BlumeKohout2017, Sugiyama2018, Lin2019}. 

If states and observables are informationally complete for SE, $V$ is the entire Hermitian-matrix space of SE, and LOT is the same as GST applied on SE. Without temporal correlations, the states and observables are usually informationally complete for the system, i.e.~$V$ is the entire Hermitian-matrix space of the system, and LOT is the same as GST applied on the system. In general, $V$ is neither the entire space of SE nor the entire space of the system, in which case LOT is different from GST. 

In the exact LOT protocol introduced in this section, we have assumed that the dimension of the subspace $V$ is known. If we can collect all the data generated by the operation set $O$, we can find out the dimension of $V$ by analysing the number of linearly independent states (in the subspace $V_{\rm out}$). In Sec.~\ref{sec:appQT}, we introduce approximate LOT protocols for the practical implementation, in which we do not need to assume that the dimension of $V$ is known. 

\section{Space dimension truncation}
\label{sec:truncation}

Usually, the environment is a high-dimensional Hilbert space. Although only the subspace $V$ is relevant in the exact LOT, its dimension could still be too high to allow the exact LOT to be implemented. Therefore, a practical LOT protocol is approximate and requires that a low-dimensional subspace approximately characterize the quantum computer. Here, we give a sufficient condition for the existence of such a subspace.

We consider the case that $d < d_V$, i.e.~states $\{ \rket{\rho_i} ~\vert~ i = 1,\ldots,d \}$ and observables $\{ \rbra{Q_k} ~\vert~ k = 1,\ldots,d  \}$ are not sufficient for implementing the exact LOT. We define a quantum computer to be approximately characterized by $\Pi_{\rm in} M_{\rm in}$, $M_{\rm out} \Pi_{\rm in}$ and $\{ \Pi_{\rm in} \mathcal{O}(\chi) \Pi_{\rm in} \}$ if the subspace spanned by $\{ \rket{\rho_i}\}$ is approximately invariant under operations $\{ \mathcal{O}(\chi) \}$. Here, $\Pi_{\rm in}$ and $\Pi_{\rm out}$ are orthogonal projections on subspaces ${\rm span} ( \{ \rket{\rho_i} \} )$ and ${\rm span} ( \{ \rbra{Q_k} \} )$, respectively.

If $\norm{\rket{\rho_i}} \leq N_{\rho}$, $\norm{\rbra{Q_k}} \leq N_Q$, $\norm{\mathcal{O}(\chi)} \leq N_\mathcal{O}$, and $\norm{ \Pi_{\rm in} \mathcal{O}(\chi) \Pi_{\rm in} - \mathcal{O}(\chi) \Pi_{\rm in} } \leq \epsilon$~\cite{footnote}, we have
\begin{eqnarray}
&& \Vert M_{\rm out} \mathcal{O}(\chi_N) \cdots \mathcal{O}(\chi_1) M_{\rm in} \notag \\
&& - M_{\rm out} \Pi_{\rm in}\mathcal{O}(\chi_N)\Pi_{\rm in} \cdots \Pi_{\rm in}\mathcal{O}(\chi_1)\Pi_{\rm in} M_{\rm in} \Vert_{\rm max} \notag \\
&\leq & N_Q N_\rho \left[ \left(N_\mathcal{O} + \epsilon\right)^N - N_\mathcal{O}^N \right] \sim N_Q N_\rho N_\mathcal{O}^{N-1} \times N\epsilon ~~
\end{eqnarray}
for an arbitrary sequence of elementary operations. See Appendix~\ref{app:truncation} for the proof. Here, we always have $N_\mathcal{O} = 1$ by taking the trace norm. A small $\epsilon$ means that the subspace is approximately invariant under elementary operations, in which case an approximate tomography is possible.

There are various ways to find out the truncated dimension $d < d_V$. For low-frequency noise and context-dependent noise, the dimension can be determined by the prior knowledge about the noise, as we show in Sec.~\ref{sec:app_model}. We can validate the truncation by computing the spectrum of the Gram matrix $g^{\rm t}$, see Sec.~\ref{sec:LIM}. Similarly, for operations with temporal correlations, the number of eigenvalues in the spectrum of an operation is more than $d_{\rm S}^2$, where $d_{\rm S}$ is the dimension of the Hilbert space of the system. The spectrum of an operation can be measured using the spectral quantum tomography~\cite{Helsen2019}, which can also be used to determine the truncated dimension. 

\section{Approximate models of temporally correlated errors}
\label{sec:app_model}

The practical use of LOT requires that a low-dimensional approximate model exists. In this section, we consider two typical sources of temporal correlations, i.e.~low-frequency noise and context-dependent noise. For both of them, low-dimensional approximate models exist.

\subsection{Low-frequency noise and classical random variables}
\label{sec:low-frequency}

A typical source of temporally correlated errors in laboratory systems is the stochastic variation of classical parameters as discussed in Sec.~\ref{sec:model}. For instances, in the trapped ion system, drifts of laser parameters cause time dependent coherent error~\cite{Wineland1998}; in the superconducting qubit system, fluctuations in the quasiparticle population lead to temporal variations in the qubit decay rate~\cite{Gustavsson2016}. If the correlation time of the noise is negligible compared with the time of a quantum gate, temporal correlations in gate errors are insignificant. However, if the correlation time is comparable or even longer than the gate time, errors are correlated, i.e.~a sequence of quantum gates cannot be factorized because of low-frequency noise. We show that errors with this kind of correlations can be efficiently approximated using a low-dimensional model if moments of the parameter distribution converge rapidly. We remark that LOT is not limited to classical correlations and can be applied to general cases as long as the dimension truncation is valid. 

As the same as in Sec.~\ref{sec:model}, We can use
\begin{eqnarray}
\rho = \int d\lambda p(\lambda) \rho_{\rm S}(\lambda) \otimes \ketbra{\lambda}{\lambda}_{\rm E}
\label{eq:LF_rho}
\end{eqnarray}
to describe a state that depends on random variables $\lambda$. Here, $\lambda$ is an array with $n$ elements that respectively denote $n$ variables, and $\rho_{\rm S}(\lambda)$ is the state of the system when variables are $\lambda$. An operation depending on random variables reads
\begin{eqnarray}
\mathcal{O}(\chi) = \mathcal{T}(\chi) \int d\lambda \mathcal{O}_{\rm S}(\chi,\lambda) \otimes [\ketbra{\lambda}{\lambda}_{\rm E}],
\label{eq:LF_O}
\end{eqnarray}
where $\mathcal{O}_{\rm S}(\chi,\lambda)$ is the operation on the system when variables are $\lambda$. Compared with the expression in Sec.~\ref{sec:model}, there is an additional operation $\mathcal{T}(\chi)$ in $\mathcal{O}(\chi)$, which describes the stochastic evolution of variables in the time of the operation. Here $\mathcal{T}(\chi) = [\openone_{\rm S}] \otimes \mathcal{T}_{\rm E}(\chi)$, $\mathcal{T}_{\rm E}(\chi) = \int d\lambda' d\lambda t_{\lambda',\lambda}(\chi) [\ketbra{\lambda'}{\lambda}_{\rm E}]$, $t_{\lambda',\lambda}(\chi)$ is the transition probability density from $\lambda$ to $\lambda'$, and $\openone$ is the identity operator. Similarly, an observable depending on random variables reads
\begin{eqnarray}
Q = \int d\lambda Q_{\rm S}(\lambda) \otimes \ketbra{\lambda}{\lambda}_{\rm E},
\label{eq:LF_Q}
\end{eqnarray}
where $Q_{\rm S}(\lambda)$ is the observable of the system when variables are $\lambda$. 

The approximate model is given by
\begin{eqnarray}
\rho^{\rm a} &=& \sum_{\lambda \in L} p^{\rm a}(\lambda) \rho_{\rm S}(\lambda) \otimes \ketbra{\lambda}{\lambda}^{\rm a}_{\rm E}, \label{eq:LF_rho_a} \\
\mathcal{O}^{\rm a}(\chi) &=& \mathcal{T}^{\rm a}(\chi) \sum_{\lambda \in L} \mathcal{O}_{\rm S}(\chi,\lambda) \otimes [\ketbra{\lambda}{\lambda}^{\rm a}_{\rm E}], \label{eq:LF_O_a} \\
Q^{\rm a} &=& \sum_{\lambda \in L} Q_{\rm S}(\lambda) \otimes \ketbra{\lambda}{\lambda}^{\rm a}_{\rm E}. \label{eq:LF_Q_a}
\end{eqnarray}
Here, $L$ is a finite subset of random variables. If $\lambda$ takes $m$ values in $L$, i.e.~$\abs{L} = m$, the environment in the approximate model is $m$-dimensional. The transition operation in the approximate model is $\mathcal{T}^{\rm a}(\chi) = [\openone_{\rm S}] \otimes \mathcal{T}^{\rm a}_{\rm E}(\chi)$, where $\mathcal{T}^{\rm a}_{\rm E}(\chi) = \sum_{\lambda',\lambda \in L} t^{\rm a}_{\lambda',\lambda}(\chi) [\ketbra{\lambda'}{\lambda}^{\rm a}_{\rm E}]$. We remark that $\rho_{\rm S}(\lambda)$, $\mathcal{O}_{\rm S}(\chi,\lambda)$ and $Q_{\rm S}(\lambda)$ are the same as in Eqs.~(\ref{eq:LF_rho})-(\ref{eq:LF_Q}). By properly choosing the subset of random variables $L$, the distribution $p^{\rm a}(\lambda)$ and transition matrices $t^{\rm a}_{\lambda',\lambda}(\chi)$, such a model can approximately characterize the quantum computer as we will show next. 

We focus on the case of only one random variable, and the generalization to the case of multiple variables is straightforward. In a quantum computation platform, the effect of the noise on quantum operations should be weak, i.e.~error rates are low. In this case, only low-order moments are important. Using the Taylor expansion, we have
\begin{eqnarray}
\rho_{\rm S}(\lambda) &=& \sum_{l=0}^{\infty} \rho_{\rm S}^{(l)} \lambda^l, \\
\mathcal{O}_{\rm S}(\chi,\lambda) &=& \sum_{l=0}^{\infty} \mathcal{O}_{\rm S}^{(l)}(\chi) \lambda^l, \\
Q_{\rm S}(\lambda) &=& \sum_{l=0}^{\infty} Q_{\rm S}^{(l)} \lambda^l.
\end{eqnarray}
Then, the quantum computation of a mean value, i.e.~the mean of an observable $Q$ in the state $\rho$ after a sequence of operations, can be expressed as
\begin{eqnarray}
&& \rbra{Q} \mathcal{O}(\chi_N) \cdots \mathcal{O}(\chi_1) \rket{\rho} \notag \\
&=& \sum_{\bf l} \rbra{Q_{\rm S}^{(l_{N+1})}} \mathcal{O}_{\rm S}^{(l_N)}(\chi_N) \cdots \mathcal{O}_{\rm S}^{(l_1)}(\chi_1) \rket{\rho_{\rm S}^{(l_0)}} \notag \\
&&\times \overline{\lambda_{N+1}^{l_{N+1}} \lambda_N^{l_N} \cdots \lambda_1^{l_1} \lambda_0^{l_0}},
\end{eqnarray}
where ${\bf l} = (l_0,l_1,\ldots,l_N,l_{N+1})$, $\lambda_j$ is the value of the variable at the time of the $j^{\rm th}$ operation, and the overline denotes the average. Therefore, the behaviour of the quantum computer is determined by correlations of random variables. If these correlations can be approximately reconstructed in the approximate model, the model approximately characterizes the quantum computer. 

Correlations are formally defined here. We introduce the operator $\hat{\lambda} = \int d\lambda \ketbra{\lambda}{\lambda}_{\rm E}$. Then,
\begin{eqnarray}
&& \overline{\lambda_{N+1}^{l_{N+1}} \lambda_N^{l_N} \cdots \lambda_1^{l_1} \lambda_0^{l_0}} \notag \\
&=& \Tr\left[ \hat{\lambda}^{l_{N+1}} \mathcal{T}_{\rm E}(\chi_N) \hat{\lambda}^{l_{N}} \cdots \mathcal{T}_{\rm E}(\chi_1) \hat{\lambda}^{l_{1}} \hat{\lambda}^{l_{0}} \rho_{\rm E} \right],
\end{eqnarray}
where $\rho_{\rm E} = \int d\lambda p(\lambda) \ketbra{\lambda}{\lambda}_{\rm E}$. 

\subsubsection{Second-order approximation}

First, we consider the case that the contribution of high-order correlations other than $\overline{\lambda_j}$ and $\overline{\lambda_{j'} \lambda_j}$ is negligible, the distribution of $\lambda$ is stationary, and the correlation time is much longer than the time scale of a quantum circuit. In this case, only $\overline{\lambda_j}$ and $\overline{\lambda_{j'} \lambda_j}$ are important. Without loss of generality, we assume that the distribution is centered at $\lambda = 0$, i.e.~$\overline{\lambda_j} = 0$. Because of the long correlation time, the change of the random variable is slow, and $\overline{\lambda_{j'} \lambda_j} \simeq \sigma^2$ is approximately a constant. Such correlations can be reconstructed in the approximate model with $m = 2$. We can take parameters in the approximate model as $L = \{ \lambda = \pm \sigma \}$, $p^{\rm a}(\pm \sigma) = 1/2$ and $t^{\rm a}_{\lambda',\lambda}(\chi) = \delta_{\lambda',\lambda}$. 

Next, we consider the case that the correlation is not a constant but decreases with time. We assume that for each operation $\chi$ the correlation is reduced by a factor of $e^{\Gamma(\chi)}$. If the correlation decreases exponentially with time, $\Gamma(\chi)$ is proportional to the operation time. The correlation reads $\overline{\lambda_{j'} \lambda_j} = \sigma^2 \exp[-\sum_{i = \max\{1,j\}}^{j'-1} \Gamma(\chi_i)]$. This two-time correlation can also be reconstructed in the approximate model with $m = 2$. The only difference is the transition matrix. We take the transition matrix as $\mathcal{T}^{\rm a}_{\rm E}(\chi) = \mathcal{T}'_{\rm E}(\Gamma(\chi))$, and
\begin{eqnarray}
\mathcal{T}'_{\rm E}(\Gamma) &\equiv & \frac{1+e^{-\Gamma}}{2} ( [\ketbra{+}{+}^{\rm a}_{\rm E}] + [\ketbra{-}{-}^{\rm a}_{\rm E}] ) \notag \\
&& + \frac{1-e^{-\Gamma}}{2} ( [\ketbra{-}{+}^{\rm a}_{\rm E}] + [\ketbra{+}{-}^{\rm a}_{\rm E}] ).
\end{eqnarray}
where $\ket{\pm}^{\rm a}_{\rm E} \equiv \ket{\lambda = \pm \sigma}^{\rm a}_{\rm E}$. Then we have $\mathcal{T}'_{\rm E}(\Gamma')\mathcal{T}'_{\rm E}(\Gamma) = \mathcal{T}'_{\rm E}(\Gamma'+\Gamma)$. 

\subsubsection{High-order approximations and multiple variables}

We consider the case that the change of the random variable is negligible in the time scale of a quantum circuit, i.e.~$\mathcal{T}_{\rm E}(\chi) \simeq [\openone_{\rm E}]$. Then, correlations become $\overline{\lambda_{N+1}^{l_{N+1}} \lambda_N^{l_N} \cdots \lambda_1^{l_1} \lambda_0^{l_0}} \simeq \int d\lambda \lambda^l$, where $l = \sum_{j=0}^{N+1} l_j$. If the contribution of correlations with $l > l_{\rm t}$ is negligible, we only need to reconstruct correlations with $l \leq l_{\rm t}$ in the approximate model, which is always possible by taking $m = \lceil (l_{\rm t}+1)/2 \rceil$~\cite{Miller1983}. We remark that $m$ is the dimension of the environment in the approximate model. 

It is similar for multiple random variables. If moments of the distribution converge rapidly, the Gaussian cubature approximation can be applied~\cite{DeVuysta2007}. Then, up to $l_{\rm t}^{\rm th}$-order moments can be reconstructed with $m = {n_\lambda+l_{\rm t} \choose l_{\rm t}}$, where $n_\lambda$ is the number of random variables. 

\subsection{Classical context-dependent noise}
\label{sec:context-dependent}

Context dependence is the effect that the error in an operation depends on previous operations, i.e.~the environment has the memory of previous operations. Here, we consider the case that the environment has a record of the classical information about previous operations. Because this kind of effects is in the scope of our general model of the quantum computer in Sec.~\ref{sec:model}, our results can be applied to the context-dependent noise. A list of previous operations is the classical information, so the memory of previous operations can be treated as a set of classical variables whose evolution is operation-dependent. Therefore, we can use the same formalism for classical random variables to characterize the context-dependent noise. We consider two examples as follows. 

In the ion trap, the temperature of ions may depend on how many gates have been performed after the last cooling operation, and the fidelity of a gate depends on the temperature. This effect can be characterized using a set of discretized variables $\lambda = (n_1,n_2,\ldots)$. Here, $n_i$ denotes the phonon number of the $i^{\rm th}$ mode. Because of the low temperature of ions, each $n_i$ can be truncated at a small number. Suppose the evolution of the qubit state mainly depends on the distribution at the beginning of the gate, the gate can be expressed as the same as in Eq.~(\ref{eq:LF_O_a}), where $\mathcal{T}(\chi)$ describes the heating process in the gate $\chi$. The cooling operation can also be expressed in this form. Then, we can apply the approximation similar to classical random variables. 

If the error in an operation only significantly depends on a few previous operations, we can use a low-dimensional environment to characterize the effect. We focus on the case that the error only depends on the last one operation, and it can be generalised to the case of depending on multiple previous operations. We characterize this effect using one discretised variable $\lambda \in \{\chi\}$, where $\{\chi\}$ is the list of all possible operations. The state after the operation $\lambda$ can be expressed in the form $\rho = \rho_{\rm S} \otimes \ketbra{\lambda}{\lambda}_{\rm E}$. An operation can be expressed in the form $\mathcal{O}(\chi) = \sum_\lambda \mathcal{O}_{\rm S}(\chi,\lambda) \otimes [\ketbra{\chi}{\lambda}_{\rm E}]$, where $\mathcal{O}_{\rm S}(\chi,\lambda)$ is the operation on the system when the last operation is $\lambda$. After the operation, the state becomes $\mathcal{O}(\chi) \rho = [\mathcal{O}_{\rm S}(\chi,\lambda) \rho_{\rm S}] \otimes \ketbra{\chi}{\chi}_{\rm E}$. We remark that it is not necessary that $\ket{\lambda}_{\rm E}$ corresponds to a pure state in a physical Hilbert space. 

\section{Approximate quantum tomography}
\label{sec:appQT}

The exact tomography protocol is not practical because of the high-dimensional state space of the environment. In Sec.~\ref{sec:app_model}, we show that an effective model with a low-dimensional environment state space can approximately characterize the quantum computer for typical temporally correlated noises. In this section, we discuss how to implement LOT to obtain a low-dimensional approximate model of the quantum computer. There are two approaches of self-consistent tomography, the linear inversion method (LIM) and the maximum likelihood estimation (MLE)~\cite{Merkel2013, BlumeKohout2013, Stark2014, Greenbaum2015, BlumeKohout2017, Sugiyama2018}, and we will discuss both of them. 

\subsection{Linear inversion method}
\label{sec:LIM}

Given sufficient data from the experiment, we can use LIM to obtain an exact model of the quantum computer as discussed in Sec.~\ref{sec:exact}. However, to obtain an approximate model, even if the approximate model exists, LIM does not always work. We suppose that $d \times d$ matrices $M^{\rm a}_{\rm in}$, $M^{\rm a}_{\rm out}$ and $\{ \mathcal{O}^{\rm a}(\chi) \}$ satisfy
\begin{eqnarray}
&& \Vert M_{\rm out} \mathcal{O}(\chi_N) \cdots \mathcal{O}(\chi_1) M_{\rm in} \notag \\
&& - M^{\rm a}_{\rm out} \mathcal{O}^{\rm a}(\chi_N) \cdots \mathcal{O}^{\rm a}(\chi_1) M^{\rm a}_{\rm in} \Vert_{\rm max}
\leq \epsilon_N
\label{eq:approximation}
\end{eqnarray}
for any sequence of elementary operations, where $\epsilon_N$ is a small quantity depending on $N$, and $d < d_V$. Then, these matrices form a model that approximately characterizes the quantum computer. If the approximate model exists, we only need to obtain $d \times d$ matrices $g^{\rm a} = M^{\rm a}_{\rm out} M^{\rm a}_{\rm in}$ and $\{ \widetilde{\mathcal{O}}^{\rm a}(\chi) = M^{\rm a}_{\rm out} \mathcal{O}^{\rm a}(\chi) M^{\rm a}_{\rm in} \}$ in order to approximately characterize the quantum computer. Because $\norm{g - g^{\rm a}}_{\rm max} \leq \epsilon_0$ and $\norm{\widetilde{\mathcal{O}}(\chi) - \widetilde{\mathcal{O}}^{\rm a}(\chi)}_{\rm max} \leq \epsilon_1$, we can directly use $g$ and $\{ \widetilde{\mathcal{O}}(\chi) \}$ as estimates of $g^{\rm a}$ and $\{ \widetilde{\mathcal{O}}^{\rm a}(\chi) \}$, which can be obtained in the experiment. However, $g^{-1}$ may be very different from $(g^{\rm a})^{-1}$, and in this case Eq.~(\ref{eq:main_GST}) may not even approximately hold. We remark that Eq.~(\ref{eq:main_GST}) always exactly holds if $d = d_V$. If Eq.~(\ref{eq:main_GST}) does not hold, LIM does not work. 

LIM works for the approximate model if the following conditions are satisfied. $\{\rket{\rho^{\rm a}_i}\}$ are columns of $M^{\rm a}_{\rm in}$, and $\{\rbra{Q^{\rm a}_k}\}$ are rows of $M^{\rm a}_{\rm out}$. Then, if $\norm{\rket{\rho^{\rm a}_i}} \leq N^{\rm a}_{\rho}$, $\norm{\rbra{Q^{\rm a}_k}} \leq N^{\rm a}_Q$, $\norm{\mathcal{O}^{\rm a}(\chi)} \leq N^{\rm a}_\mathcal{O}$, $\norm{ M_{\rm in}^{\rm a} g^{-1} M_{\rm out}^{\rm a} - \openone } \leq \epsilon_g$ and $\norm{(M_{\rm out}^{\rm a})^{-1} \widetilde{\mathcal{O}}(\chi) (M_{\rm in}^{\rm a})^{-1} - \mathcal{O}^{\rm a}(\chi)} \leq \epsilon_\mathcal{O}$, we have
\begin{eqnarray}
&& \Vert \widetilde{\mathcal{O}}(\chi_N) g^{-1} \cdots \widetilde{\mathcal{O}}(\chi_2) g^{-1} \widetilde{\mathcal{O}}(\chi_1) \notag \\
&& - M_{\rm out}^{\rm a} \mathcal{O}^{\rm a}(\chi_N) \cdots \mathcal{O}^{\rm a}(\chi_2) \mathcal{O}^{\rm a}(\chi_1) M_{\rm in}^{\rm a} \Vert_{\rm max} \notag \\
&\leq & N^{\rm a}_Q N^{\rm a}_\rho \left[ \left(1 + \epsilon_g\right)^{N-1} \left(N_\mathcal{O} + \epsilon_\mathcal{O}\right)^N - N_\mathcal{O}^N \right] \notag \\
&\sim & N^{\rm a}_Q N^{\rm a}_\rho \times N_\mathcal{O}^{N-1} \left[ (N-1)N_\mathcal{O}\epsilon_g + N\epsilon_\mathcal{O} \right]
\label{eq:approximationLIM}
\end{eqnarray}
for any sequence of elementary operations. See Appendix~\ref{app:LIM} for the proof. Therefore, LIM works under conditions that $N^{\rm a}_\mathcal{O} \lesssim 1$, and $\epsilon_g$ and $\epsilon_\mathcal{O}$ are small quantities. 

We apply this result to the approximate model given by the approximately invariant subspace $\Pi_{\rm in}$. We have
\begin{eqnarray}
&& \Vert \widetilde{\mathcal{O}}(\chi_N) g^{-1} \cdots \widetilde{\mathcal{O}}(\chi_2) g^{-1} \widetilde{\mathcal{O}}(\chi_1) \notag \\
&& - M_{\rm out} \Pi_{\rm in}\mathcal{O}(\chi_N)\Pi_{\rm in} \cdots \Pi_{\rm in}\mathcal{O}(\chi_1)\Pi_{\rm in} M_{\rm in} \Vert_{\rm max} \notag \\
&\leq & N_Q (1+\epsilon_\Pi) N_\rho \left[ \left( N_\mathcal{O} + \frac{\epsilon}{1-\epsilon_\Pi} \right)^N - \left( N_\mathcal{O} + \epsilon \right)^N \right] \notag \\
&\sim & N_Q N_\rho N_\mathcal{O}^{N-1} \times \frac{(1+\epsilon_\Pi)\epsilon_\Pi}{1-\epsilon_\Pi}N\epsilon
\end{eqnarray}
for any sequence of elementary operations. See Appendix~\ref{app:LIM} for the proof. Here, $\epsilon_\Pi = \norm{\Pi_{\rm in} - \Pi_{\rm out}}$, and we have $N_\mathcal{O} = 1$ by taking the trace norm. Therefore, if $\norm{ \Pi_{\rm in} \mathcal{O}(\chi) \Pi_{\rm in} - \mathcal{O}(\chi) \Pi_{\rm in} } \ll 1$ and $\norm{\Pi_{\rm in} - \Pi_{\rm out}} \ll 1$ for the trace norm, LIM can be applied. Here, $\norm{\Pi_{\rm in} - \Pi_{\rm out}} \ll 1$ means that two subspaces $\Pi_{\rm in}$ and $\Pi_{\rm out}$ are approximately the same. 

In order to implement LIM, we need to find states $\{\rket{\rho_i}\}$ and observables $\{\rbra{Q_k}\}$ corresponding to an approximately invariant subspace $\Pi_{\rm in}$. For this purpose, we choose a set of trial states $\{\rket{\rho^{\rm t}_i}\}$ and a set of trial observables $\{\rbra{Q^{\rm t}_k}\}$. The most interesting approximate invariant subspace is the subspace containing the initial state $\rket{\rho_{\rm in}}$. If such an approximate invariant subspace exists, states in the form $\mathcal{O}(\chi_N) \cdots \mathcal{O}(\chi_1) \rket{\rho_{\rm in}}$ are all approximately within the subspace, as long as $N$ is sufficiently small. Therefore, we can choose the initial state $\rket{\rho_{\rm in}}$ and states in the form $\mathcal{O}(\chi_N) \cdots \mathcal{O}(\chi_1) \rket{\rho_{\rm in}}$ as trial states. Similarly, we can choose the observable $\rbra{Q_{\rm out}}$ and effective observables in the form $\rbra{Q_{\rm out}} \mathcal{O}(\chi_N) \cdots \mathcal{O}(\chi_1)$ as trial observables. 

In the ideal case, i.e.~$\Pi_{\rm in}$ is an exactly invariant subspace, and trial states and observables are exactly within $\Pi_{\rm in}$, i.e.~$\Pi_{\rm in} \rket{\rho^{\rm t}_i} = \rket{\rho^{\rm t}_i}$ and $\rbra{Q^{\rm t}_k} \Pi_{\rm in} = \rbra{Q^{\rm t}_k}$. Then, the rank of $g^{\rm t} = M^{\rm t}_{\rm out} M^{\rm t}_{\rm in}$ is not greater than the dimension of the subspace $\Pi_{\rm in}$. Here, $M^{\rm t}_{\rm in}$ and $M^{\rm t}_{\rm out}$ are matrices corresponding to trial states and observables, respectively. If the subspace $\Pi_{\rm in}$ is approximately invariant, the matrix $g^{\rm t}$ should still be close to a matrix with a rank not greater than the subspace dimension. Therefore, we can determine $M_{\rm in}$ and $M_{\rm out}$ by performing a truncation on the spectrum of singular values of $g^{\rm t}$, i.e.~we choose $M_{\rm in}$ and $M_{\rm out}$ corresponding to the greatest $d$ singular values of $g^{\rm t}$. Suppose the singular value decomposition is $Ug^{\rm t}V = \Lambda$, where $\Lambda = {\rm diag}(s_1, s_2, \ldots, s_{d^{\rm t}})$ and $s_1 \geq s_2 \geq \cdots \geq s_{d^{\rm t}}$, then we use states $\{\rket{\rho_i} = \sum_j \rket{\rho^{\rm t}_i} V_{i,j} ~\vert~ i = 1,\ldots d \}$ and observables $\{\rbra{Q_k} = \sum_j U_{k,j} \rbra{Q^{\rm t}_j} ~\vert~ k = 1,\ldots d \}$ to implement LOT. 

The approximate LOT protocol using LIM is as follows: 
\begin{itemize}
\item[$\bullet$] Choose a set of states $\{ \rket{\rho^{\rm t}_i} = \mathcal{O}_i \rket{\rho_{\rm in}} ~\vert~ \mathcal{O}_i\in O; i = 1,\cdots,d^{\rm t}\}$ and a set of observables $\{ \rbra{Q^{\rm t}_k} = \rbra{Q_{\rm out}} \mathcal{O}_k' ~\vert~ \mathcal{O}_k'\in O; k = 1,\cdots,d^{\rm t} \}$. We always take $\rket{\rho^{\rm t}_1} = \rket{\rho_{\rm in}}$ and $\rbra{Q^{\rm t}_1} = \rbra{Q_{\rm out}}$. Here, $O$ is the set of operation sequences. 
\item[$\bullet$] Obtain matrices $g^{\rm t} = M^{\rm t}_{\rm out} M^{\rm t}_{\rm in}$ and $\widetilde{\mathcal{O}}^{\rm t}(\chi) = M^{\rm t}_{\rm out} \mathcal{O}(\chi) M^{\rm t}_{\rm in}$ for each $\chi$ in the experiment. Here, $M^{\rm t}_{\rm in} = [ \; \rket{\rho^{\rm t}_1} \;\cdots\; \rket{\rho^{\rm t}_{d^{\rm t}}} \; ]$ and $M^{\rm t}_{\rm out} = [ \; \rbra{Q^{\rm t}_1}^{\rm T} \;\cdots\; \rbra{Q^{\rm t}_{d^{\rm t}}}^{\rm T} \; ]^{\rm T}$. 
\item[$\bullet$] Compute the singular value decomposition $Ug^{\rm t}V = \Lambda$, where $\Lambda = {\rm diag}(s_1, s_2, \ldots, s_{d^{\rm t}})$, and singular values are sorted in the descending order $s_1 \geq s_2 \geq \cdots \geq s_{d^{\rm t}}$. 
\item[$\bullet$] Choose the dimension $d$. Compute $g = {\rm diag}(s_1, s_2, \ldots, s_d) = D \Lambda D^\dag$ and $\widetilde{\mathcal{O}}(\chi) = D U \widetilde{\mathcal{O}}^{\rm t}(\chi) V D^\dag$ for each $\chi$. Here, $D$ is a $d\times d^{\rm t}$ matrix, and $D_{i,i'} = \delta_{i,i'}$. 
\item[$\bullet$] Choose a $d$-dimensional invertible real matrix $\widehat{M}_{\rm in}$, and compute $\widehat{M}_{\rm out} = g\widehat{M}_{\rm in}^{-1}$. 
\item[$\bullet$] Compute $\rket{\widehat{\rho}_{\rm in}} = \sum_{i=1}^d \widehat{M}_{\rm in; \bullet,i} V^*_{1,i} = \widehat{M}_{\rm out}^{-1} D U g^{\rm t}_{\bullet,1}$, $\rbra{\widehat{Q}_{\rm out}} = \sum_{k=1}^d U^*_{k,1} \widehat{M}_{\rm out; k,\bullet} = g^{\rm t}_{1,\bullet} V D^\dag \widehat{M}_{\rm in}^{-1}$, and $\widehat{\mathcal{O}}(\chi) = \widehat{M}_{\rm out}^{-1} \widetilde{\mathcal{O}}(\chi) \widehat{M}_{\rm in}^{-1}$ for each $\chi$. 
\end{itemize} 

\subsection{Maximum likelihood estimation}

\begin{figure}[tbp]
\centering
\includegraphics[width=1\linewidth]{\figpath 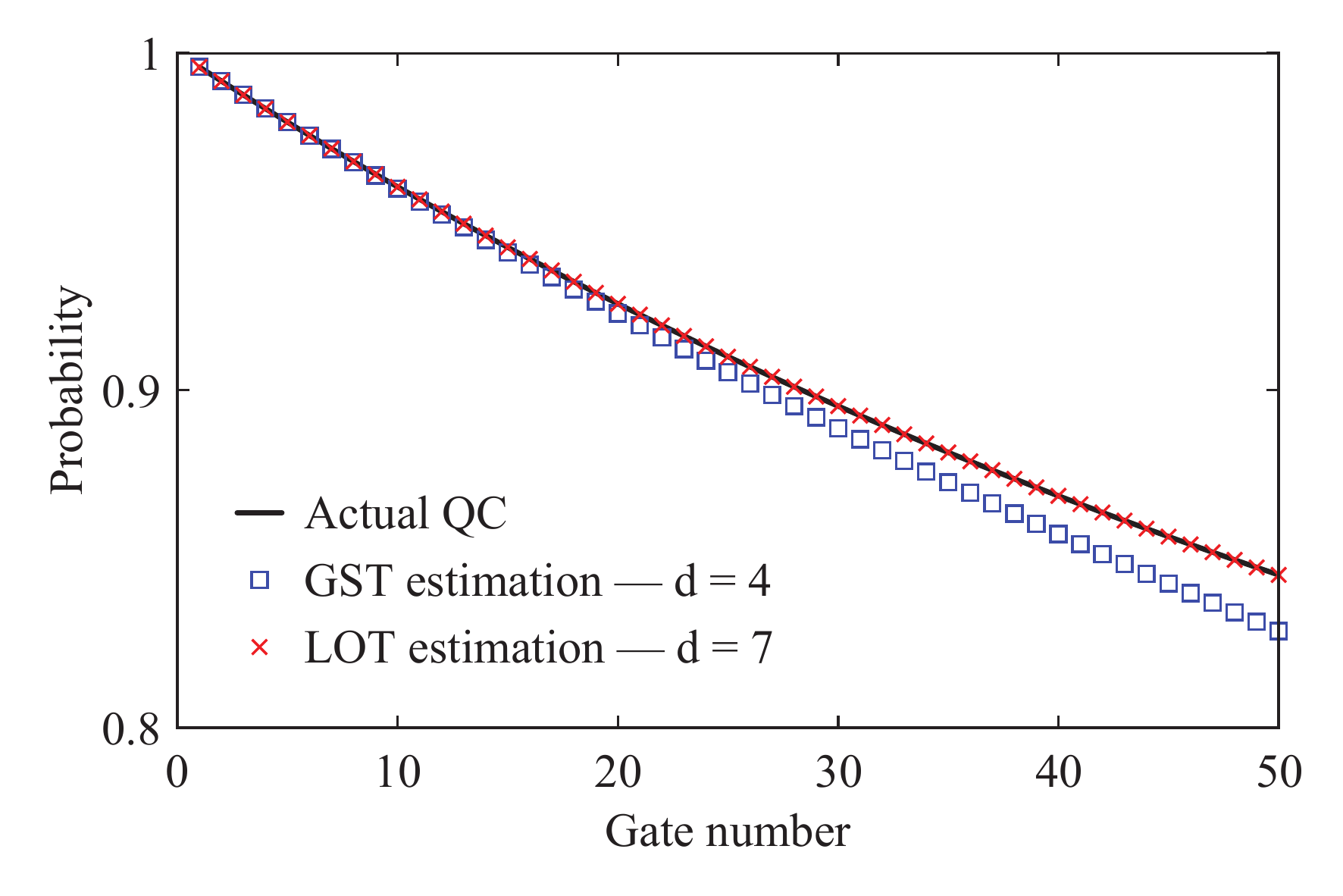}
\caption{
Probabilities in the state $\ket{0}$ after a sequence of randomly chosen Hadamard and phase gates as functions of the gate number. We initialise the qubit in the state $\ket{0}$, perform the random gate sequence and measure the probability in the state $\ket{0}$. We only take into account gate sequences that the final state is $\ket{0}$ in the case of ideal gates without error. Therefore the probability should be $1$ in this case. In our simulation, we take $\sigma = 1$ and $\eta = 0.02$. In the presence of errors, the probability in the actual quantum computation (QC) decreases with the gate number (black curve). Based on error models obtained in linear operator tomography (LOT) using maximum likelihood estimation, we can estimate the decreasing probability, and the results are plotted. We can find the that the error model with $d = 7$ (red crosses) fits the actual behavior of the quantum computer much more accurately than the error model with $d = 4$ (blue squares). When $d = 4$, LOT is equivalent to conventional gate set tomography (GST). Results for the linear inversion method are similar. See Appendix~\ref{app:simulation} for details of the simulation and more data. 
}
\label{fig:plot}
\end{figure}

The alternative method for determining the error model is based on MLE. Given a model of the quantum computer with unknown parameters, MLE is to find the estimated values of the unknown parameters, such that the likelihood of samples observed in the experiment is maximized. Let $d$-dimensional column vector $\rket{\bar{\rho}_{\rm in}({\bf x})}$, row vector $\rbra{\bar{Q}_{\rm out}({\bf x})}$ and matrices $\{ \bar{\mathcal{O}}(\chi,{\bf x}) \}$, respectively corresponding to the initial state, measured observable and operations, be the theoretical model of the quantum computer depending on parameters ${\bf x}$. Our goal is to estimate parameters ${\bf x}$ based on data from the experiment. The mean of $Q_{\rm out}$ in $\rho_{\rm in}$ after a sequence of operations measured in the experiment is $C = \rbra{Q_{\rm out}} \mathcal{O}(\chi_N) \cdots \mathcal{O}(\chi_1) \rket{\rho_{\rm in}} + \delta$, where $\delta$ is the deviation from the actual mean value, and the mean according to the model is $\bar{C}({\bf x}) = \rbra{\bar{Q}_{\rm out}({\bf x})} \bar{\mathcal{O}}(\chi_N,{\bf x}) \cdots \bar{\mathcal{O}}(\chi_1,{\bf x}) \rket{\bar{\rho}_{\rm in}({\bf x})}$. Using the Gaussian approximation, the likelihood function to be maximised is $L({\bf x}) = \exp\{ -[\bar{C}({\bf x}) - C]^2/\sigma^2 \}$, where $\sigma$ is the standard deviation of $C$. In the practical implementation, multiple quantum circuits and corresponding mean values are needed to determine the error model. The protocol is as follows: 
\begin{itemize}
\item[$\bullet$] Parameterize the $d$-dimensional column vector $\rket{\bar{\rho}_{\rm in}({\bf x})}$, row vector $\rbra{\bar{Q}_{\rm out}({\bf x})}$ and matrix $\bar{\mathcal{O}}(\chi,{\bf x})$ for each $\chi$ as functions of parameters ${\bf x} = ( x_1, x_2, \ldots )$.
\item[$\bullet$] Choose $M$ circuits $\{ \boldsymbol{\chi}_1, \ldots, \boldsymbol{\chi}_M \}$. For each circuit $\boldsymbol{\chi}_m = (\chi_{m,1}, \cdots, \chi_{m,N_m})$, obtain $\rbra{Q_{\rm out}} \mathcal{O}(\chi_{m,N_m}) \cdots \mathcal{O}(\chi_{m,1}) \rket{\rho_{\rm in}}$ in the experiment. The result is $C_m$.
\item[$\bullet$] Minimise the likelihood function $L({\bf x}) = \prod_{m=1}^M \exp\{ -[\bar{C}_m({\bf x}) - C_m]^2/\sigma_m^2 \}$, where $\bar{C}_m({\bf x}) = \rbra{\bar{Q}_{\rm out}({\bf x})} \bar{\mathcal{O}}(\chi_{m,N_m},{\bf x}) \cdots \bar{\mathcal{O}}(\chi_{m,1},{\bf x}) \rket{\bar{\rho}_{\rm in}({\bf x})}$, and $\sigma_m^2$ is the variance of $C_m$. The likelihood function is minimised at $\widehat{\bf x} = \text{arg min}_{\bf x} \{ L({\bf x}) \}$.
\item[$\bullet$] Compute $\rket{\widehat{\rho}_{\rm in}} = \rket{\bar{\rho}_{\rm in}(\widehat{\bf x})}$, $\rbra{\widehat{Q}_{\rm out}} = \rbra{\bar{Q}_{\rm out}(\widehat{\bf x})}$, and $\widehat{\mathcal{O}}(\chi) = \bar{\mathcal{O}}(\chi,\widehat{\bf x})$ for each $\chi$.
\end{itemize}
We can parameterize the error model by taking each vector and matrix element as a parameter. If the main source of temporal correlations is low-frequency noise or context-dependent noise as discussed in Sec.~\ref{sec:app_model}, we can parameterize the error model according to Eqs.~(\ref{eq:LF_rho_a})-(\ref{eq:LF_Q_a}). 

\section{Numerical simulation of low-frequency noise}
\label{sec:simulation}

\begin{figure}[tbp]
\centering
\includegraphics[width=1\linewidth]{\figpath 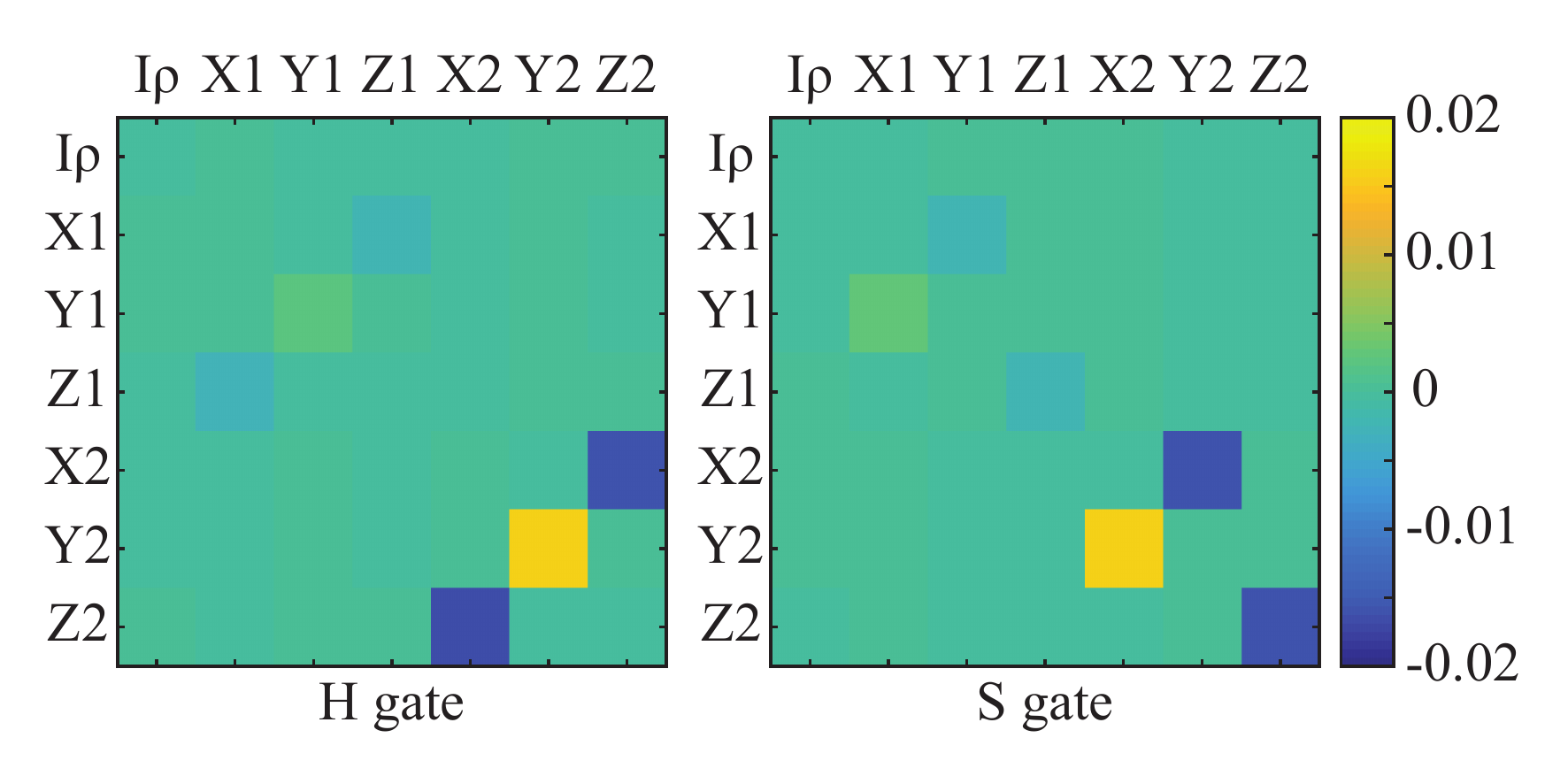}
\caption{
Pauli transfer matrices obtained using the linear inversion method. The difference between the matrix obtained in tomography and the matrix of the ideal gate, i.e.~$M_{\mathcal{O}}-M_{\mathcal{O}}^{\rm ideal}$, is plotted. See Appendix~\ref{app:simulation} for $M_{\mathcal{O}}^{\rm ideal}$. 
}
\label{fig:matrix}
\end{figure}

To demonstrate our protocols numerically, we consider a model of one qubit with time-dependent gate fidelities and implement LOT using the numerical simulation on a classical computer. In the model, gate fidelities depend on a low-frequency time-dependent variable $\lambda$, whose distribution is Gaussian. We assume that the change of the variable is negligible in the time scale of a quantum circuit. The initial state and the observable to be measured are error free, which are $\rho_{\rm in} = \ketbra{0}{0}_{\rm S}\otimes \rho_{\rm E}$ and $Q_{\rm out} = \ketbra{0}{0}_{\rm S}\otimes \openone_{\rm E}$, respectively. Here, the state of the environment is $\rho_{\rm E} = \int d\lambda \frac{1}{2\pi\sigma^2}\exp(\frac{\lambda^2}{2\sigma^2})\ketbra{\lambda}{\lambda}_{\rm E}$. We remark that LOT can deal with state preparation and measurement errors as the same as GST. We neglect state preparation and measurement errors in our simulation for simplification. Errors in single-qubit gates are depolarizing errors, and depolarizing rates depend on $\lambda$. For a unitary single-qubit gate $G$, the actual gate with error is $\mathcal{O}_{\rm S}(G, \lambda) = \mathcal{E}(\epsilon_G(\lambda))[G]$, where $\epsilon_G(\lambda)$ is the depolarizing rate, $\mathcal{E}(\epsilon) = (1-\epsilon)[\openone_{\rm S}] + \epsilon\mathcal{D}$, and $\mathcal{D} = \frac{1}{4}([\openone_{\rm S}] + [X_{\rm S}] + [Y_{\rm S}] + [Z_{\rm S}])$. Here, $X$, $Y$ and $Z$ are Pauli operators. Then the operation on SE for the gate $G$ is $\mathcal{O}(G) = \int d\lambda \mathcal{O}_{\rm S}(G, \lambda) \otimes \ketbra{\lambda}{\lambda}_{\rm E}$.

We consider two single-qubit gates, the Hadamard gate $H$ and the phase gate $S$, which can generate all single-qubit Clifford gates. We take $\epsilon_H(\lambda) = \epsilon_S(\lambda) = \eta[1 - \exp(-\lambda^2)]$, therefore, two gates are both optimised at $\lambda = 0$. Here, $\eta \in [0,1]$ denotes the strength of the noise. RB is the usual way of the verification of a quantum computing system~\cite{Emerson2005, Knill2008, Wallman2014}. In our simulation, we perform a sequence of $H$ and $S$ gates randomly chosen in the uniform distribution. We initialize the qubit in the state $\ket{0}$, perform the random gate sequence and measure the probability in the state $\ket{0}$. We only take into account gate sequences that the final state is $\ket{0}$ in the case of ideal gates without error, so that the probability in the state $\ket{0}$ is expected to be $1$. When errors are switched on, the probability in the state $\ket{0}$ is $F(N_{\rm g}) = (1 + 1/\sqrt{1+2N_{\rm g}\sigma^2})/2$ if $\eta = 1$, where $N_{\rm g}$ is the number of gates in the random gate sequence. The non-exponential decay of the probability is due to temporal correlations~\cite{Emerson2005, Knill2008, Magesan2011, Wallman2014, Fogarty2015, Ball2016, Mavadia2018}. Without temporal correlation, the probability decreases exponentially with the gate number. If depolarizing rates are constants, i.e.~$\epsilon_H(\lambda) = \epsilon_S(\lambda) = \epsilon$, we have $[1+(1-\epsilon)^{N_{\rm g}}]/2$.

In our simulation, we implement both LIM and MLE. We take the dimension of the state space $d = 4,7$ to compare LOT with conventional GST. In approximate models of classical random variables with stationary distribution (see Sec.~\ref{sec:low-frequency}), the state space is $[(d_{\rm S}^2-1)m+1]$-dimensional when the system and environment Hilbert spaces are respectively $d_{\rm S}$-dimensional and $m$-dimensional, as explained in Appendix~\ref{app:dimension}. Therefore, $d=4,7$ correspond to $m=1,2$ approximations, respectively. If $d = 4$, the LOT protocol is equivalent to conventional GST protocol, because the one-dimensional environment is trivial and does not have any effect. As shown in Fig.~\ref{fig:plot}, LOT with $d = 7$ can characterize the behavior of the quantum computer much more accurately than LOT with $d = 4$ (i.e.~conventional GST).

In the simulation of LOT using MLE, we parametrize the state, observable and operations as follows. The state is in the form $\bar{\rho}_{\rm in} = \sum_{\lambda \in L} p^{\rm a}(\lambda) \ketbra{0}{0}_{\rm S} \otimes \ketbra{\lambda}{\lambda}^{\rm a}_{\rm E}$. The observable is in the form $\bar{Q}_{\rm out} = \sum_{\lambda \in L} \ketbra{0}{0}_{\rm S} \otimes \ketbra{\lambda}{\lambda}^{\rm a}_{\rm E}$. The gate $G$ with error is in the form $\bar{\mathcal{O}}(G) = \sum_{\lambda \in L} \mathcal{E}(\epsilon^{\rm a}_G(\lambda)) [G] \otimes \ketbra{\lambda}{\lambda}^{\rm a}_{\rm E}$. We take $\{ p^{\rm a}(\lambda) \}$ and $\{ \epsilon^{\rm a}_G(\lambda) \}$ as parameters (i.e.~${\bf x}$) in MLE. With the error model parametrized in this way, the number of values that $\lambda$ can take is important, but the value of $\lambda$ is not important. For the one-dimensional environment approximation, i.e.~$d = 4$, we take $L = \{1\}$; and for the two-dimensional environment approximation, i.e.~$d = 7$, we take $L = \{1,2\}$. Using MLE, we obtain $p^{\rm a}(1) = 0.5606$, $p^{\rm a}(2) = 0.4394$, $\epsilon^{\rm a}_H(1) \simeq \epsilon^{\rm a}_S(1) = 2.485\times10^{-3}$ and $\epsilon^{\rm a}_H(2) \simeq \epsilon^{\rm a}_S(2) = 1.606\times10^{-2}$. 

In the case that the random variable takes two values $\lambda = 1,2$, the environment state is $\rho^{\rm a}_{\rm E} = p^{\rm a}(1) \ketbra{1}{1}^{\rm a}_{\rm E} + p^{\rm a}(2) \ketbra{2}{2}^{\rm a}_{\rm E}$. Because the distribution is stationary, the state of the environment does not evolve. Usually, we can use an eight-dimensional Pauli transfer matrix to represent an operation on a qubit and a classical bit~\cite{BlumeKohout2017}. However, because the component $I_{\rm S}\otimes \left[ p^{\rm a}(2) \ketbra{1}{1}^{\rm a}_{\rm E} - p^{\rm a}(1) \ketbra{2}{2}^{\rm a}_{\rm E} \right]$ is always zero (see Appendix~\ref{app:dimension}), the dimension of the state space is effectively seven. For the operation $\mathcal{O}$, the corresponding seven-dimensional Pauli transfer matrix is $M_{\mathcal{O};\sigma,\tau} = \Tr[\sigma\mathcal{O}(\tau)]/2$, where $\sigma,\tau = I_{\rm S}\otimes \rho^{\rm a}_{\rm E}/\sqrt{a},X_{\rm S}\otimes \ketbra{1}{1}^{\rm a}_{\rm E},Y_{\rm S}\otimes \ketbra{1}{1}^{\rm a}_{\rm E},Z_{\rm S}\otimes \ketbra{1}{1}^{\rm a}_{\rm E},X_{\rm S}\otimes \ketbra{2}{2}^{\rm a}_{\rm E},Y_{\rm S}\otimes \ketbra{2}{2}^{\rm a}_{\rm E},Z_{\rm S}\otimes \ketbra{2}{2}^{\rm a}_{\rm E}$ and $a = \Tr(\rho^{{\rm a}2}_{\rm E})$. Pauli transfer matrices obtained using LIM are show in Fig.~\ref{fig:matrix}. 

\section{Conclusions}

We have proposed self-consistent tomography protocols to obtain the model of temporally correlated errors in a quantum computer. Given sufficient data from the experiment, the model obtained in our protocol can be exact. We also propose approximate models for the practical implementation. To obtain approximate models characterizing temporal correlations, more quantities need to be measured compared with conventional QPT and GST, but the overhead is moderate. We can use such approximate models to predict the behavior of a quantum computer much more accurately than the model obtained in GST, for systems with temporally correlated errors. Our protocols provide a way to quantitatively assess temporal correlations in quantum computers. 

\begin{acknowledgments}
This work was supported by the National Key R\&D Program of China (Grant No. 2016YFA0301200) and the National Basic Research Program of China (Grant No. 2014CB921403). It is also supported by Science Challenge Project (Grant No. TZ2017003) and the National Natural Science Foundation of China (Grants No. 11774024, No. 11534002, and No. U1530401).
YL is supported by National Natural Science Foundation of China (Grant No. 11875050) and NSAF (Grant No. U1730449).
\end{acknowledgments}

\appendix

\section{Exact linear operator tomography}
\label{app:exact}

We consider two subspaces $\hat{V}_{\rm in} = {\rm span} ( \{ \rbra{Q_{\rm out}} \mathcal{O} P_{\rm in} ~\vert~ \mathcal{O} \in O \} )$ and $\hat{V}_{\rm out} = {\rm span} ( \{ P_{\rm out} \mathcal{O} \rket{\rho_{\rm in}} ~\vert~ \mathcal{O} \in O \} )$. We use $\hat{P}_{\rm in}$ and $\hat{P}_{\rm out}$ to denote the orthogonal projection on $\hat{V}_{\rm in}$ and $\hat{V}_{\rm out}$, respectively. Here, $\hat{V}_{\rm out} = V$ and $\hat{P}_{\rm out} = P$. $P_{{\rm in}\cap \overline{\rm out}}$ is the orthogonal projection on the intersection of $V_{\rm in}$ and the orthogonal complement of $V_{\rm out}$. $P_{{\rm out}\cap \overline{\rm in}}$ is the orthogonal projection on the intersection of $V_{\rm out}$ and the orthogonal complement of $V_{\rm in}$. Then, $\hat{P}_{\rm in} = P_{\rm in} - P_{{\rm in}\cap \overline{\rm out}}$ and $\hat{P}_{\rm out} = P_{\rm out} - P_{{\rm out}\cap \overline{\rm in}}$.


\begin{lemma}
Let $\mathcal{O} \in O$, all the following expressions are valid.
\begin{eqnarray}
\hat{P}_{\rm in} &=& \hat{P}_{\rm in}P_{\rm in} = P_{\rm in}\hat{P}_{\rm in}, \label{eq:in} \\
\hat{P}_{\rm out} &=& \hat{P}_{\rm out}P_{\rm out} = P_{\rm out}\hat{P}_{\rm out}, \label{eq:out} \\
P_{\rm out}P_{\rm in} &=& P_{\rm out}\hat{P}_{\rm in} = \hat{P}_{\rm out}P_{\rm in}, \label{eq:out_in} \\
P_{\rm out} \mathcal{O} P_{\rm in} &=& P_{\rm out} \mathcal{O} \hat{P}_{\rm in} = P_{\rm out} \hat{P}_{\rm in} \mathcal{O} \hat{P}_{\rm in} \notag \\
&=& \hat{P}_{\rm out} \mathcal{O} P_{\rm in} = \hat{P}_{\rm out} \mathcal{O} \hat{P}_{\rm out} P_{\rm in}. \label{eq:out_O_in}
\end{eqnarray}
\label{lemma1}
\end{lemma}

\begin{proof}
We have
\begin{eqnarray}
P_{\rm in}P_{{\rm in}\cap \overline{\rm out}} &=& P_{{\rm in}\cap \overline{\rm out}}P_{\rm in} = P_{{\rm in}\cap \overline{\rm out}}, \\
P_{\rm out}P_{{\rm in}\cap \overline{\rm out}} &=& P_{{\rm in}\cap \overline{\rm out}}P_{\rm out} = 0, \\
P_{\rm out}P_{{\rm out}\cap \overline{\rm in}} &=& P_{{\rm out}\cap \overline{\rm in}}P_{\rm out} = P_{{\rm out}\cap \overline{\rm in}}, \\
P_{\rm in}P_{{\rm out}\cap \overline{\rm in}} &=& P_{{\rm out}\cap \overline{\rm in}}P_{\rm in} = 0.
\end{eqnarray}
Therefore, Eq.~(\ref{eq:in}), Eq.~(\ref{eq:out}) and Eq.~(\ref{eq:out_in}) are valid.

Using $\mathcal{O} P_{\rm in} = P_{\rm in} \mathcal{O} P_{\rm in}$ and $P_{\rm out} \mathcal{O} = P_{\rm out} \mathcal{O} P_{\rm out}$, we have
\begin{eqnarray}
P_{\rm out} \mathcal{O} P_{\rm in} &=& P_{\rm out} \mathcal{O} P_{\rm out} P_{\rm in} = P_{\rm out} \mathcal{O} P_{\rm out} \hat{P}_{\rm in} \notag \\
&=& P_{\rm out} \mathcal{O} \hat{P}_{\rm in} = P_{\rm out} \mathcal{O} P_{\rm in} \hat{P}_{\rm in} \notag \\
&=& P_{\rm out} P_{\rm in} \mathcal{O} P_{\rm in} \hat{P}_{\rm in} = P_{\rm out} \hat{P}_{\rm in} \mathcal{O} P_{\rm in} \hat{P}_{\rm in} \notag \\
&=& P_{\rm out} \hat{P}_{\rm in} \mathcal{O} \hat{P}_{\rm in}.
\end{eqnarray}
Similarly,
\begin{eqnarray}
P_{\rm out} \mathcal{O} P_{\rm in} &=& P_{\rm out} P_{\rm in} \mathcal{O} P_{\rm in} = \hat{P}_{\rm out} P_{\rm in} \mathcal{O} P_{\rm in} \notag \\
&=& \hat{P}_{\rm out} \mathcal{O} P_{\rm in} = \hat{P}_{\rm out} P_{\rm out} \mathcal{O} P_{\rm in} \notag \\
&=& \hat{P}_{\rm out} P_{\rm out} \mathcal{O} P_{\rm out} P_{\rm in} = \hat{P}_{\rm out} P_{\rm out} \mathcal{O} \hat{P}_{\rm out} P_{\rm in} \notag \\
&=& \hat{P}_{\rm out} \mathcal{O} \hat{P}_{\rm out} P_{\rm in}.
\end{eqnarray}
Therefore, Eq.~(\ref{eq:out_O_in}) is valid.
\end{proof}


\begin{theorem}
Let $\rket{\sigma} \in V_{\rm in}$, $\rbra{H} \in V_{\rm out}$ and $\mathcal{O}_i \in O$, where $i = 1,2,\ldots,N$. Then,
\begin{eqnarray}
&& \rbra{H} \mathcal{O}_N \cdots \mathcal{O}_2 \mathcal{O}_1 \rket{\sigma} \notag \\
&=& \rbra{H} \hat{P}_{\rm in} \mathcal{O}_N \hat{P}_{\rm in} \cdots \hat{P}_{\rm in} \mathcal{O}_2 \hat{P}_{\rm in} \mathcal{O}_1 \hat{P}_{\rm in} \rket{\sigma} \notag \\
&=& \rbra{H} \hat{P}_{\rm out} \mathcal{O}_N \hat{P}_{\rm out} \cdots \hat{P}_{\rm out} \mathcal{O}_2 \hat{P}_{\rm out} \mathcal{O}_1 \hat{P}_{\rm out} \rket{\sigma}. ~~~
\end{eqnarray}
\label{theorem1}
\end{theorem}

\begin{proof}
Using Lemma~\ref{lemma1}, we have
\begin{eqnarray}
P_{\rm out} \mathcal{O}' P_{\rm out} \mathcal{O} \hat{P}_{\rm in} &=& P_{\rm out} \mathcal{O}' P_{\rm out} \mathcal{O} P_{\rm in} \hat{P}_{\rm in} \notag \\
&=& P_{\rm out} \mathcal{O}' P_{\rm out} \hat{P}_{\rm in} \mathcal{O} \hat{P}_{\rm in} \hat{P}_{\rm in} \notag \\
&=& P_{\rm out} \mathcal{O}' \hat{P}_{\rm in} \mathcal{O} \hat{P}_{\rm in},
\end{eqnarray}
and
\begin{eqnarray}
\hat{P}_{\rm out} \mathcal{O} P_{\rm in} \mathcal{O} P_{\rm in} &=& \hat{P}_{\rm out} P_{\rm out} \mathcal{O} P_{\rm in} \mathcal{O}' P_{\rm in} \notag \\
&=& \hat{P}_{\rm out} \hat{P}_{\rm out} \mathcal{O} \hat{P}_{\rm out} P_{\rm in} \mathcal{O}' P_{\rm in} \notag \\
&=& \hat{P}_{\rm out} \mathcal{O} \hat{P}_{\rm out} \mathcal{O}' P_{\rm in}.
\end{eqnarray}

Using $\rket{\sigma} = P_{\rm in} \rket{\sigma}$ and $\rbra{H} = \rbra{H} P_{\rm out}$, we have
\begin{eqnarray}
&& \rbra{H} \mathcal{O}_N \cdots \mathcal{O}_2 \mathcal{O}_1 \rket{\sigma} \notag \\
&=& \rbra{H} P_{\rm out} \mathcal{O}_N \cdots \mathcal{O}_2 \mathcal{O}_1 P_{\rm in}  \rket{\sigma} \notag \\
&=& \rbra{H} P_{\rm out} \mathcal{O}_N P_{\rm out} \cdots P_{\rm out} \mathcal{O}_2 P_{\rm out} \mathcal{O}_1 P_{\rm in}  \rket{\sigma} \notag \\
&=& \rbra{H} P_{\rm out} \mathcal{O}_N P_{\rm out} \cdots P_{\rm out} \mathcal{O}_2 P_{\rm out} \mathcal{O}_1 \hat{P}_{\rm in}  \rket{\sigma} \notag \\
&=& \rbra{H} P_{\rm out} \mathcal{O}_N \hat{P}_{\rm in} \cdots \hat{P}_{\rm in} \mathcal{O}_2 \hat{P}_{\rm in} \mathcal{O}_1 \hat{P}_{\rm in}  \rket{\sigma} \notag \\
&=& \rbra{H} P_{\rm out} \hat{P}_{\rm in} \mathcal{O}_N \hat{P}_{\rm in} \cdots \hat{P}_{\rm in} \mathcal{O}_2 \hat{P}_{\rm in} \mathcal{O}_1 \hat{P}_{\rm in}  \rket{\sigma} \notag \\
&=& \rbra{H} \hat{P}_{\rm in} \mathcal{O}_N \hat{P}_{\rm in} \cdots \hat{P}_{\rm in} \mathcal{O}_2 \hat{P}_{\rm in} \mathcal{O}_1 \hat{P}_{\rm in}  \rket{\sigma}.
\end{eqnarray}
Similarly,
\begin{eqnarray}
&& \rbra{H} \mathcal{O}_N \cdots \mathcal{O}_2 \mathcal{O}_1 \rket{\sigma} \notag \\
&=& \rbra{H} P_{\rm out} \mathcal{O}_N \cdots \mathcal{O}_2 \mathcal{O}_1 P_{\rm in}  \rket{\sigma} \notag \\
&=& \rbra{H} P_{\rm out} \mathcal{O}_N P_{\rm in} \cdots P_{\rm in} \mathcal{O}_2 P_{\rm in} \mathcal{O}_1 P_{\rm in}  \rket{\sigma} \notag \\
&=& \rbra{H} \hat{P}_{\rm out} \mathcal{O}_N P_{\rm in} \cdots P_{\rm in} \mathcal{O}_2 P_{\rm in} \mathcal{O}_1 P_{\rm in}  \rket{\sigma} \notag \\
&=& \rbra{H} \hat{P}_{\rm out} \mathcal{O}_N \hat{P}_{\rm out} \cdots \hat{P}_{\rm out} \mathcal{O}_2 \hat{P}_{\rm out} \mathcal{O}_1 P_{\rm in}  \rket{\sigma} \notag \\
&=& \rbra{H} \hat{P}_{\rm out} \mathcal{O}_N \hat{P}_{\rm out} \cdots \hat{P}_{\rm out} \mathcal{O}_2 \hat{P}_{\rm out} \mathcal{O}_1 \hat{P}_{\rm out} P_{\rm in}  \rket{\sigma} \notag \\
&=& \rbra{H} \hat{P}_{\rm out} \mathcal{O}_N \hat{P}_{\rm out} \cdots \hat{P}_{\rm out} \mathcal{O}_2 \hat{P}_{\rm out} \mathcal{O}_1 \hat{P}_{\rm out}  \rket{\sigma}.
\end{eqnarray}
\end{proof}


\begin{theorem}
Let $d = \Tr (P)$, and each of $\{ P \rket{\rho_i} \}$ and $\{ \rbra{Q_k} P \}$ be a set of $d$ linearly-independent vectors. Then, $g = M_{\rm out} M_{\rm in} = M_{\rm out} P M_{\rm in}$ is invertible, $\widetilde{\mathcal{O}} = M_{\rm out} \mathcal{O} M_{\rm in} = M_{\rm out} P \mathcal{O} P M_{\rm in}$ and
\begin{eqnarray}
C &=& M_{\rm out} \mathcal{O}_N \cdots \mathcal{O}_2 \mathcal{O}_1 M_{\rm in} \notag \\
&=& M_{\rm out} P \mathcal{O}_N P \cdots \mathcal{O}_2 P \mathcal{O}_1 P M_{\rm in} \notag \\
&=& \widetilde{\mathcal{O}}_N g^{-1} \cdots \widetilde{\mathcal{O}}_2 g^{-1} \widetilde{\mathcal{O}}_1.
\label{eq:theorem2}
\end{eqnarray}
\label{theorem2}
\end{theorem}

\begin{proof}
We remark that $P = \hat{P}_{\rm out}$, and the theorem is also valid for $\hat{P}_{\rm in}$.

According to definitions of $M_{\rm out}$ and $M_{\rm in}$, we have $g_{k,i} = \rbraket{Q_k}{\rho_i}$. Because $\rket{\rho_i} \in V_{\rm in}$ and $\rbra{Q_k} \in V_{\rm out}$, we have $g_{k,i} = \rbra{Q_k} \hat{P}_{\rm in} \rket{\rho_i} = \rbra{Q_k} \hat{P}_{\rm out} \rket{\rho_i}$. Here, we have used Theorem~\ref{theorem1}. Therefore, $M_{\rm out} M_{\rm in} = M_{\rm out} P M_{\rm in}$.

Similarly, we have $\widetilde{\mathcal{O}}_{k,i} = \rbra{Q_k} \mathcal{O} \rket{\rho_i}  = \rbra{Q_k} \hat{P}_{\rm in} \mathcal{O} \hat{P}_{\rm in} \rket{\rho_i}  = \rbra{Q_k} \hat{P}_{\rm out} \mathcal{O} \hat{P}_{\rm out} \rket{\rho_i} $. Therefore, $M_{\rm out} \mathcal{O} M_{\rm in} = M_{\rm out} P \mathcal{O} P M_{\rm in}$

We also have
\begin{eqnarray}
&& C_{k,i} = \rbra{Q_k} \mathcal{O}_N \cdots \mathcal{O}_2 \mathcal{O}_1 \rket{\rho_i} \notag \\
&=& \rbra{Q_k} \hat{P}_{\rm in} \mathcal{O}_N \hat{P}_{\rm in} \cdots \hat{P}_{\rm in} \mathcal{O}_2 \hat{P}_{\rm in} \mathcal{O}_1 \hat{P}_{\rm in} \rket{\rho_i} \notag \\
&=& \rbra{Q_k} \hat{P}_{\rm out} \mathcal{O}_N \hat{P}_{\rm out} \cdots \hat{P}_{\rm out} \mathcal{O}_2 \hat{P}_{\rm out} \mathcal{O}_1 \hat{P}_{\rm out} \rket{\rho_i}.
~~~
\end{eqnarray}
Therefore, the first two lines of Eq.~(\ref{eq:theorem2}) are equal.

$P M_{\rm in}$ is a full rank $d_{\rm H}^2 \times \Tr (P)$ matrix, and $M_{\rm out} P$ is a full rank $\Tr (P) \times d_{\rm H}^2$ matrix. Thus, $(PM_{\rm in})^+ PM_{\rm in} = \openone$, $PM_{\rm in} (PM_{\rm in})^+ = P$, $M_{\rm out}P (M_{\rm out}P)^+ = \openone$ and $(M_{\rm out}P)^+ M_{\rm out}P = P$. Here, $A^+$ denotes the pseudo inverse of matrix $A$.

Using pseudo inverses, we have $g^{-1} = (PM_{\rm in})^+ (M_{\rm out}P)^+$, i.e.~$g = M_{\rm out}P PM_{\rm in}$ is invertible. Thus,
\begin{eqnarray}
&& \widetilde{\mathcal{O}}_N g^{-1} \cdots \widetilde{\mathcal{O}}_2 g^{-1} \widetilde{\mathcal{O}}_1 \notag \\
&=& M_{\rm out}P \mathcal{O}_N PM_{\rm in} (PM_{\rm in})^+ (M_{\rm out}P)^+ \cdots \notag \\
&&\times M_{\rm out}P \mathcal{O}_2 PM_{\rm in} (PM_{\rm in})^+ (M_{\rm out}P)^+ M_{\rm out}P \mathcal{O}_1 PM_{\rm in} \notag \\
&=& M_{\rm out} P \mathcal{O}_N P \cdots \mathcal{O}_2 P \mathcal{O}_1 P M_{\rm in}.
\end{eqnarray}
Therefore, the last two lines of Eq.~(\ref{eq:theorem2}) are equal.
\end{proof}

\section{Space dimension truncation}
\label{app:truncation}

We use $\norm{\cdot}$ to denote a vector norm satisfying $\abs{\rbraket{A}{B}} \leq \norm{\rbra{A}} \norm{\rket{B}}$ and the submultiplicative matrix norm induced by the vector norm, i.e.~$\norm{\mathcal{O}\rket{B}} \leq \norm{\mathcal{O}} \norm{\rket{B}}$ and $\norm{\mathcal{O}_1\mathcal{O}_2} \leq \norm{\mathcal{O}_1} \norm{\mathcal{O}_2}$.

Two examples of such norms. First, we can take $\norm{\rket{B}} = \sqrt{\rbraket{B}{B}} = \sqrt{\Tr (B^2)}$. Then, $\norm{\rket{B}} = \sqrt{\sum_i \sigma_i^2}$, where $\{ \sigma_i \}$ are singular values of $B$. Second, we can take $\norm{\rket{B}} = \norm{B}_1 = \sum_i \sigma_i \geq \sqrt{\sum_i \sigma_i^2}$, where $\norm{\cdot}_1$ denotes the trace norm. 

We use $\norm{\cdot}_{\rm max}$ to denote the max norm. 


\begin{lemma}
Let $N_Q \geq \norm{\rbra{Q_k}}$ for all $k$, and $N_\rho \geq \norm{\rket{\rho_i}}$ for all $i$. Then
\begin{eqnarray}
\norm{ M_{\rm out} \mathcal{O}_N \cdots \mathcal{O}_1 M_{\rm in} }_{\rm max} \leq N_Q N_\rho \prod_{j=1}^N \norm{\mathcal{O}_j}.
\end{eqnarray}
\label{lemma2}
\end{lemma}

According to the property of vector norm, the proof of Lemma~\ref{lemma2} is straightforward. 


\begin{theorem}
Let $N_Q \geq \norm{\rbra{Q_k}}$ for all $k$, and $N_\rho \geq \norm{\rket{\rho_i}}$ for all $i$. Then, for any sequence of operations,
\begin{eqnarray}
&& \Vert M_{\rm out} \overline{\mathcal{O}}_1 M_{\rm in} - M_{\rm out} \overline{\mathcal{O}}_{n+1} \overline{\mathcal{P}}_n M_{\rm in} \Vert_{\rm max} \notag \\
&\leq & N_Q N_\rho \prod_{j=n+1}^N \norm{\mathcal{O}_j} \notag \\
&& \times \left[ \prod_{j=1}^n \left(\norm{\mathcal{O}_j} + \norm{\delta_{\mathcal{O}_j}}\right) - \prod_{j=1}^n \norm{\mathcal{O}_j} \right],
\label{eq:theorem4}
\end{eqnarray}
where
\begin{eqnarray}
\delta_\mathcal{O} &=& \Pi_{\rm in}\mathcal{O}\Pi_{\rm in} - \mathcal{O}\Pi_{\rm in}, \\
\overline{\mathcal{O}}_m &=& \mathcal{O}_N \cdots \mathcal{O}_{m+1} \mathcal{O}_m, \\
\overline{\mathcal{P}}_m &=& \Pi_{\rm in} \mathcal{O}_m \Pi_{\rm in} \cdots \Pi_{\rm in} \mathcal{O}_2 \Pi_{\rm in} \mathcal{O}_1 \Pi_{\rm in}.
\end{eqnarray}
\label{theorem4}
\end{theorem}

\begin{proof}
Inequality~(\ref{eq:theorem4}) holds for $n = 0$, because
\begin{eqnarray}
\Vert M_{\rm out} \overline{\mathcal{O}}_1 M_{\rm in} - M_{\rm out} \overline{\mathcal{O}}_1 \Pi_{\rm in} M_{\rm in} \Vert_{\rm max} = 0.
\end{eqnarray}

If inequality~(\ref{eq:theorem4}) holds for $n = m$, then it also holds for $n = m+1$. Now we assume that inequality~(\ref{eq:theorem4}) holds for $n = m$. Because
\begin{eqnarray}
&& \Vert M_{\rm out} \overline{\mathcal{O}}_1 M_{\rm in} \Vert_{\rm max} \leq N_Q N_\rho \prod_{j=1}^N \norm{\mathcal{O}_j},
\end{eqnarray}
we have
\begin{eqnarray}
&& \Vert M_{\rm out} \overline{\mathcal{O}}_{m+1} \overline{\mathcal{P}}_m M_{\rm in} \Vert_{\rm max} \notag \\
&\leq & N_Q N_\rho \prod_{j=m+1}^N \norm{\mathcal{O}_j} \prod_{j=1}^m \left(\norm{\mathcal{O}_j} + \norm{\delta_{\mathcal{O}_j}}\right),
\end{eqnarray}
Because
\begin{eqnarray}
&& M_{\rm out} \overline{\mathcal{O}}_{m+2} \overline{\mathcal{P}}_{m+1} M_{\rm in} \notag \\
&=& M_{\rm out} \overline{\mathcal{O}}_{m+2} \delta_{\mathcal{O}_{m+1}} \overline{\mathcal{P}}_m M_{\rm in} \notag \\
&& + M_{\rm out} \overline{\mathcal{O}}_{m+1} \overline{\mathcal{P}}_m M_{\rm in},
\end{eqnarray}
we have
\begin{eqnarray}
&& \Vert M_{\rm out} \overline{\mathcal{O}}_1 M_{\rm in} - M_{\rm out} \overline{\mathcal{O}}_{m+2} \overline{\mathcal{P}}_{m+1} M_{\rm in} \Vert_{\rm max} \notag \\
&\leq & N_Q N_\rho \prod_{j=m+2}^N \norm{\mathcal{O}_j} \norm{\delta_{\mathcal{O}_{m+1}}} \prod_{j=1}^m \left(\norm{\mathcal{O}_j} + \norm{\delta_{\mathcal{O}_j}}\right) \notag \\
&& + N_Q N_\rho  \prod_{j=m+1}^N \norm{\mathcal{O}_j} \notag \\
&& \times \left[ \prod_{j=1}^m \left(\norm{\mathcal{O}_j} + \norm{\delta_{\mathcal{O}_j}}\right) - \prod_{j=1}^m \norm{\mathcal{O}_j} \right] \notag \\
&=& N_Q N_\rho  \prod_{j=m+2}^N \norm{\mathcal{O}_j} \notag \\
&& \times \left[ \prod_{j=1}^{m+1} \left(\norm{\mathcal{O}_j} + \norm{\delta_{\mathcal{O}_j}}\right) - \prod_{j=1}^{m+1} \norm{\mathcal{O}_j} \right],
\end{eqnarray}
i.e.~inequality~(\ref{eq:theorem4}) holds for $n = m+1$.
\end{proof}

\section{Linear inversion method}
\label{app:LIM}


\begin{theorem}
$\{\rket{\rho^{\rm a}_i}\}$ are columns of $M^{\rm a}_{\rm in}$, and $\{\rbra{Q^{\rm a}_k}\}$ are rows of $M^{\rm a}_{\rm out}$. Let $N^{\rm a}_Q \geq \norm{\rbra{Q_k^{\rm a}}}$ for all $k$, and $N^{\rm a}_\rho \geq \norm{\rket{\rho^{\rm a}_i}}$ for all $i$. If $g$, $M_{\rm in}^{\rm a}$ and $M_{\rm out}^{\rm a}$ are inevitable, for any sequence of operations in $\{\mathcal{O}(\chi)\}$,
\begin{eqnarray}
&& \Vert \widetilde{\mathcal{O}}(\chi_N) g^{-1} \cdots \widetilde{\mathcal{O}}(\chi_2) g^{-1} \widetilde{\mathcal{O}}(\chi_1) \notag \\
&& - M_{\rm out}^{\rm a} \mathcal{O}^{\rm a}(\chi_N) \cdots \mathcal{O}^{\rm a}(\chi_2) \mathcal{O}^{\rm a}(\chi_1) M_{\rm in}^{\rm a} \Vert_{\rm max} \notag \\
&\leq & N^{\rm a}_Q N^{\rm a}_\rho
\left[ \left(1 + \norm{\delta_g}\right)^{N-1} \prod_{j=1}^N \left(\norm{\mathcal{O}^{\rm a}(\chi_j)} + \norm{\delta_{\chi_j}}\right) \right. \notag \\
&& - \left. \prod_{j=1}^N \norm{\mathcal{O}^{\rm a}(\chi_j)} \right],
\label{eq:theorem5}
\end{eqnarray}
where
\begin{eqnarray}
\delta_g &=& M_{\rm in}^{\rm a} g^{-1} M_{\rm out}^{\rm a} - \openone, \\
\delta_\chi &=& (M_{\rm out}^{\rm a})^{-1} \widetilde{\mathcal{O}}(\chi) (M_{\rm in}^{\rm a})^{-1} - \mathcal{O}^{\rm a}(\chi).
\end{eqnarray}
\label{theorem5}
\end{theorem}

\begin{proof}
We have
\begin{eqnarray}
g^{-1} &=& (M_{\rm in}^{\rm a})^{-1} (\openone + \delta_g) (M_{\rm out}^{\rm a})^{-1}, \\
\widetilde{\mathcal{O}}(\chi) &=& M_{\rm out}^{\rm a} [\mathcal{O}^{\rm a}(\chi) + \delta_\chi] M_{\rm in}^{\rm a}.
\end{eqnarray}
Then,
\begin{eqnarray}
&& \widetilde{\mathcal{O}}(\chi_N) g^{-1} \cdots \widetilde{\mathcal{O}}(\chi_2) g^{-1} \widetilde{\mathcal{O}}(\chi_1) \notag \\
&=& M_{\rm out}^{\rm a} [\mathcal{O}^{\rm a}(\chi_N) + \delta_{\chi_N}] (\openone + \delta_g) \cdots \notag \\
&& \times [\mathcal{O}^{\rm a}(\chi_2) + \delta_{\chi_2}] (\openone + \delta_g) [\mathcal{O}^{\rm a}(\chi_1) + \delta_{\chi_1}] M_{\rm in}^{\rm a}. ~~~
\end{eqnarray}
Therefore, inequality~(\ref{eq:theorem5}) holds. 
\end{proof}

We now apply Theorem~\ref{theorem5} to the approximate model given by the approximate invariant subspace $\Pi_{\rm in}$. Let $\{\rket{l} ~\vert~ l=1,2,\ldots,d\}$ be an orthonormal basis of the subspace $\Pi_{\rm in}$, i.e.~$\rbraket{l_1}{l_2} = \delta_{l_1,l_2}$ and $\Pi_{\rm in} = \sum_{l=1}^d \rketbra{l}{l}$, and $\{\rket{l^{\rm a}} ~\vert~ l=1,2,\ldots,d\}$ be an orthonormal basis of the approximate-model space, i.e.~$\rbraket{l^{\rm a}_1}{l^{\rm a}_2} = \delta_{l_1,l_2}$ and $\openone = \sum_{l=1}^d \rketbra{l^{\rm a}}{l^{\rm a}}$. Then, $T \equiv \sum_{l=1}^d \rketbra{l^{\rm a}}{l}$ is the transformation from the actual space to the approximate-model space, and $T^+ = \sum_{l=1}^d \rketbra{l}{l^{\rm a}}$ is the inverse transformation. We have $TT^+ = \openone$ and $T^+T = \Pi_{\rm in}$. The approximate model is given by
\begin{eqnarray}
M_{\rm in}^{\rm a}  &=& T M_{\rm in}, \\
M_{\rm out}^{\rm a} &=& M_{\rm out} T^+, \\
\mathcal{O}^{\rm a}(\chi) &=& T \mathcal{O}(\chi) T^+.
\end{eqnarray}

We take the vector norm in the approximate-model space $\norm{\rbra{A^{\rm a}}} = \norm{\rbra{A^{\rm a}} T}$ and $\norm{\rket{B^{\rm a}}} = \norm{T^+ \rket{B^{\rm a}}}$. Then, $\abs{\rbraket{A^{\rm a}}{B^{\rm a}}} = \norm{\rbra{A^{\rm a}}} \norm{\rket{B^{\rm a}}}$ is satisfied. We have
\begin{eqnarray}
\norm{\mathcal{O}^{\rm a} \rket{B^{\rm a}}} = \norm{T^+\mathcal{O}^{\rm a}T T^+\rket{B^{\rm a}}}.
\end{eqnarray}
Because
\begin{eqnarray}
\norm{T^+\mathcal{O}^{\rm a}T T^+\rket{B^{\rm a}}} &\leq & \norm{T^+\mathcal{O}^{\rm a}T} \norm{T^+\rket{B^{\rm a}}} \notag \\
&=& \norm{T^+\mathcal{O}^{\rm a}T} \norm{\rket{B^{\rm a}}},
\end{eqnarray}
we have
\begin{eqnarray}
\norm{\mathcal{O}^{\rm a}} \leq \norm{T^+\mathcal{O}^{\rm a}T}.
\end{eqnarray}
Let $\mathcal{O}$ be an operation satisfying $\Pi_{\rm in}\mathcal{O}\Pi_{\rm in} = T^+\mathcal{O}^{\rm a}T$, we have
\begin{eqnarray}
&& \norm{T^+\mathcal{O}^{\rm a}T T^+\rket{B^{\rm a}}} \notag \\
&=& \norm{\Pi_{\rm in}\mathcal{O}\Pi_{\rm in} T^+\rket{B^{\rm a}}} \notag \\
&\leq & \norm{\mathcal{O}\Pi_{\rm in} T^+\rket{B^{\rm a}}} + \norm{\delta_\mathcal{O} T^+\rket{B^{\rm a}}} \notag \\
&=& \norm{\mathcal{O} T^+\rket{B^{\rm a}}} + \norm{\delta_\mathcal{O} T^+\rket{B^{\rm a}}} \notag \\
&\leq & (\norm{\mathcal{O}} + \norm{\delta_\mathcal{O}}) \norm{T^+\rket{B^{\rm a}}} \notag \\
&=& (\norm{\mathcal{O}} + \norm{\delta_\mathcal{O}}) \norm{\rket{B^{\rm a}}},
\end{eqnarray}
where
\begin{eqnarray}
\delta_\mathcal{O} &=& \Pi_{\rm in}\mathcal{O}\Pi_{\rm in} - \mathcal{O}\Pi_{\rm in}.
\end{eqnarray}
Therefore, $\norm{\mathcal{O}^{\rm a}} \leq \norm{\mathcal{O}} + \norm{\delta_\mathcal{O}}$. Because $T^+\mathcal{O}^{\rm a}(\chi)T = \Pi_{\rm in}\mathcal{O}(\chi)\Pi_{\rm in}$, we have $\norm{\mathcal{O}^{\rm a}(\chi)} \leq \norm{\mathcal{O}(\chi)} + \norm{\delta_{\mathcal{O}(\chi)}}$.

We have,
\begin{eqnarray}
\delta_g &=& T M_{\rm in} g^{-1} M_{\rm out} T^+ - \openone \notag \\
&=& T (M_{\rm in} g^{-1} M_{\rm out} \Pi_{\rm in} - \Pi_{\rm in}) T^+.
\end{eqnarray}
Because $g$ is invertible, $M_{\rm in}$ and $M_{\rm out}$ are full rank. Thus, $M_{\rm in}^+ M_{\rm in} = \openone$, $M_{\rm in} M_{\rm in}^+ = \Pi_{\rm in}$, $M_{\rm out} M_{\rm out}^+ = \openone$ and $M_{\rm out}^+ M_{\rm out} = \Pi_{\rm out}$. Then,
\begin{eqnarray}
&& M_{\rm in} g^{-1} M_{\rm out} \Pi_{\rm in} = M_{\rm in} g^{-1} M_{\rm out} M_{\rm in} M_{\rm in}^+ \notag \\
&=& M_{\rm in} g^{-1} g M_{\rm in}^+ = M_{\rm in} M_{\rm in}^+ = \Pi_{\rm in}.
\end{eqnarray}
Therefore, $\delta_g = 0$.

Because $g = M_{\rm out} \Pi_{\rm in} M_{\rm in}$ is invertible, $M_{\rm out} \Pi_{\rm in}$ is full rank. Thus, $M_{\rm out} \Pi_{\rm in} (M_{\rm out} \Pi_{\rm in})^+ = \openone$ and $(M_{\rm out} \Pi_{\rm in})^+ M_{\rm out} \Pi_{\rm in} = \Pi_{\rm in}$. Then, we have $(M_{\rm out} T^+)^{-1} = T (M_{\rm out} \Pi_{\rm in})^+$ and $(T M_{\rm in})^{-1} = M_{\rm in}^+ T^+$. Therefore,
\begin{eqnarray}
\delta_\chi &=& (M_{\rm out} T^+)^{-1} \widetilde{\mathcal{O}}(\chi) (T M_{\rm in})^{-1} - T \mathcal{O}(\chi) T^+ \notag \\
&=& T \left[ (M_{\rm out} \Pi_{\rm in})^+ \widetilde{\mathcal{O}}(\chi) M_{\rm in}^+ - \Pi_{\rm in} \mathcal{O}(\chi) \Pi_{\rm in} \right] T^+. ~~~~~~
\end{eqnarray}
We have
\begin{eqnarray}
&& (M_{\rm out} \Pi_{\rm in})^+ \widetilde{\mathcal{O}}(\chi) M_{\rm in}^+ \notag \\
&=& (M_{\rm out} \Pi_{\rm in})^+ M_{\rm out} \mathcal{O}(\chi) M_{\rm in} M_{\rm in}^+ \notag \\
&=& (M_{\rm out} \Pi_{\rm in})^+ M_{\rm out} \mathcal{O}(\chi) \Pi_{\rm in} \notag \\
&=& (M_{\rm out} \Pi_{\rm in})^+ M_{\rm out} \Pi_{\rm in} \mathcal{O}(\chi) \Pi_{\rm in} - (M_{\rm out} \Pi_{\rm in})^+ M_{\rm out} \delta_{\mathcal{O}(\chi)} \notag \\
&=& \Pi_{\rm in} \mathcal{O}(\chi) \Pi_{\rm in} - (M_{\rm out} \Pi_{\rm in})^+ M_{\rm out} \delta_{\mathcal{O}(\chi)}.
\end{eqnarray}
Then,
\begin{eqnarray}
\delta_\chi = - T (M_{\rm out} \Pi_{\rm in})^+ M_{\rm out} \delta_{\mathcal{O}(\chi)} T^+.
\end{eqnarray}

Let $G = \Pi_{\rm in} (M_{\rm out} \Pi_{\rm in})^+ M_{\rm out}$, we have
\begin{eqnarray}
\Pi_{\rm in}G &=& G, \\
G\Pi_{\rm out} &=& G, \\
G\Pi_{\rm in} &=& \Pi_{\rm in} (M_{\rm out} \Pi_{\rm in})^+ M_{\rm out} \Pi_{\rm in} = \Pi_{\rm in}, \\
\Pi_{\rm out}G &=& \Pi_{\rm out} \Pi_{\rm in} (M_{\rm out} \Pi_{\rm in})^+ M_{\rm out} \notag \\
&=& M_{\rm out}^+ M_{\rm out} \Pi_{\rm in} (M_{\rm out} \Pi_{\rm in})^+ M_{\rm out} \notag \\
&=& M_{\rm out}^+ M_{\rm out} = \Pi_{\rm out}.
\end{eqnarray}
Then,
\begin{eqnarray}
(\openone - \Pi_{\rm in} + \Pi_{\rm out})G = \Pi_{\rm out}.
\end{eqnarray}
We define $\delta_\Pi \equiv \Pi_{\rm in} - \Pi_{\rm out}$. If $\openone - \delta_\Pi$ is invertible, we have
\begin{eqnarray}
G = (\openone - \delta_\Pi)^{-1} \Pi_{\rm out}.
\end{eqnarray}
Then,
\begin{eqnarray}
T^+ \delta_\chi T &=& - \Pi_{\rm in} (M_{\rm out} \Pi_{\rm in})^+ M_{\rm out} \delta_{\mathcal{O}(\chi)} \Pi_{\rm in} \notag \\
&=& - G \delta_{\mathcal{O}(\chi)} \notag \\
&=& - (\openone - \delta_\Pi)^{-1} \Pi_{\rm out} \delta_{\mathcal{O}(\chi)} \notag \\
&=& (\openone - \delta_\Pi)^{-1} \delta_\Pi \delta_{\mathcal{O}(\chi)} \notag \\
&& - (\openone - \delta_\Pi)^{-1} \Pi_{\rm in} \delta_{\mathcal{O}(\chi)} \notag \\
&=& (\openone - \delta_\Pi)^{-1} \delta_\Pi \delta_{\mathcal{O}(\chi)}.
\end{eqnarray}
Therefore,
\begin{eqnarray}
\norm{\delta_\chi} &\leq & \norm{T^+ \delta_\chi T} = \Vert (\openone - \delta_\Pi)^{-1} \delta_\Pi \delta_{\mathcal{O}(\chi)} \Vert \notag \\
&\leq & (1-\norm{\delta_\Pi})^{-1} \norm{\delta_\Pi} \norm{\delta_{\mathcal{O}(\chi)}}.
\end{eqnarray}

Using inequality~(\ref{eq:theorem5}), we have
\begin{eqnarray}
&& \Vert \widetilde{\mathcal{O}}(\chi_N) g^{-1} \cdots \widetilde{\mathcal{O}}(\chi_2) g^{-1} \widetilde{\mathcal{O}}(\chi_1) \notag \\
&& - M_{\rm out}^{\rm a} \mathcal{O}^{\rm a}(\chi_N) \cdots \mathcal{O}^{\rm a}(\chi_2) \mathcal{O}^{\rm a}(\chi_1) M_{\rm in}^{\rm a} \Vert_{\rm max} \notag \\
&\leq & N'_Q N'_\rho \left\{ \prod_{j=1}^N \left[ \norm{\mathcal{O}(\chi)} + \norm{\delta_{\mathcal{O}(\chi)}} \right. \right. \notag \\
&& \left. + (1-\norm{\delta_\Pi})^{-1} \norm{\delta_\Pi} \norm{\delta_{\mathcal{O}(\chi)}} \right] \notag \\
&& - \left. \prod_{j=1}^N \left( \norm{\mathcal{O}(\chi)} + \norm{\delta_{\mathcal{O}(\chi)}} \right) \right\},
\end{eqnarray}
where
\begin{eqnarray}
N'_Q &=& \max \{ \norm{\rbra{Q_k^{\rm a}}} \} = \max \{ \norm{\rbra{Q_k^{\rm a}}T} \} \notag \\
&=& \max \{ \norm{\rbra{Q_k}\Pi_{\rm in}} \} \notag \\
&\leq & \max \{ \norm{\rbra{Q_k}\Pi_{\rm out}} + \norm{\rbra{Q_k}\delta_\Pi} \} \notag \\
&\leq & \max \{ \norm{\rbra{Q_k}} (1 + \norm{\delta_\Pi}) \}, \\
N'_\rho &=& \max \{ \norm{\rket{\rho_i^{\rm a}}} \} = \max \{ \norm{T^+\rket{\rho_i^{\rm a}}} \} \notag \\
&=& \max \{ \norm{\Pi_{\rm in} \rket{\rho_i}} \} = \max \{ \norm{\rket{\rho_i}} \}.
\end{eqnarray}

\section{Vector space dimensions}
\label{app:dimension}

If Hilbert spaces of the system and environment are respectively $d_{\rm S}$-dimensional and $d_{\rm E}$-dimensional, the Hilbert space of SE is $(d_{\rm S}d_{\rm E})$-dimensional. Then, a column vector $\rket{\rho}$ representing the state of SE is $(d_{\rm S}^2d_{\rm E}^2)$-dimensional. We remark that $d_{\rm E} = m$.

For the classical random variable noise, the state is in the form $\rho = \sum_\lambda p(\lambda) \rho_{\rm S}(\lambda) \otimes \ketbra{\lambda}{\lambda}_{\rm E}$, i.e.~the state of the environment (in the reduced density matrix form) only has diagonal elements. Therefore, we can use a $(d_{\rm S}^2d_{\rm E})$-dimensional vector to represent the state, i.e.~take $\rket{\rho} = \sum_\lambda p(\lambda) \rket{\rho_{\rm S}(\lambda)}_{\rm S} \otimes \rket{\lambda}_{\rm E}$, where $\{ \rket{\rho_{\rm S}(\lambda)}_{\rm S} \}$ are column vectors representing states of the system, and $\{ \rket{\lambda}_{\rm E} \}$ are column vectors representing states of the environment.

The state of the system can always be expressed in the form $\rho_{\rm S} = d_{\rm S}^{-1}\openone_{\rm S} + \rho_{\rm S}'$, where $\Tr (\rho_{\rm S}') = 0$. Then $\rket{\rho_{\rm S}}_{\rm S} = \rket{\openone}_{\rm S} + \rket{\rho_{\rm S}'}_{\rm S}$, where $\rket{\openone}_{\rm S}$ represents the maximally mixed state $d_{\rm S}^{-1}\openone_{\rm S}$, and $\rket{\openone}_{\rm S}$ and $\rket{\rho_{\rm S}'}_{\rm S}$ are orthogonal. We focus on the $d_{\rm E}$-dimensional subspace ${\rm span}(\{ \rket{\openone}_{\rm S} \otimes \rket{\lambda}_{\rm E} \})$. The orthogonal projection on this subspace is $P_{\openone} = \rketbra{\openone}{\openone}_{\rm S} \otimes \rketbra{\lambda}{\lambda}_{\rm E}$. Then, $P_{\openone} \rket{\rho} = \sum_\lambda p(\lambda) \rket{\openone}_{\rm S} \otimes \rket{\lambda}_{\rm E}$. If the distribution of $\lambda$ is stationary, i.e.~$\{ p(\lambda) \}$ are invariant under operations, we have $P_{\openone} \mathcal{O} \rket{\rho} = P_{\openone} \rket{\rho}$ for any operation $\mathcal{O}$ that does not change the distribution. Therefore, if the distribution of $\lambda$ is stationary, $P_{\openone} \rket{\rho}$ is the only non-trivial vector in the subspace $P_{\openone}$ that contributes to the state, and $\rket{\rho}$ is effectively $[(d_{\rm S}^2-1)d_{\rm E}+1]$-dimensional.

\section{Details of the numerical simulation}
\label{app:simulation}

\subsection{Actual quantum computer simulation}

To simulate the behaviour of the actual quantum computer, we use the Gaussian cubature approximation to match up to the $9^{\rm th}$-order moment, by taking $e^{-\lambda^2}$ instead of $\lambda$ as the random variable. Reducing the precision of the approximation and only matching up to the $7^{\rm th}$-order moment, we find that the difference is negligible.

\subsection{Linear inversion method simulation}

Trial states and observables are generated as follows: We selected a set of gate sequences, $(G_1,G_2,\ldots,G_N)$, where $G_j = H,S$; Each gate sequence corresponds to a state $G_N\cdots G_2G_1\ket{0}$ and an observable $(G_1G_2\cdots G_N)^\dag Z(G_1G_2\cdots G_N)$. These states and observables are realised using $H$ and $S$ gates accordingly. 

When $d = 4$, four gates sequences are used in the simulation, $({\rm null})$, $(H)$, $(H,S)$ and $(H,S,H)$. The four states are $\ket{0}$, $\ket{+} = H\ket{0}$, $\ket{y+} = SH\ket{0}$ and $\ket{y-} = HSH\ket{0}$; The four observables are $Z$, $X = H^\dag Z H$, $-Y = S^\dag H^\dag YHS$ and $Y = H^\dag S^\dag H^\dag YHSH$. 

When $d = 7$, $123$ gates sequences are used in the simulation, including all sequences with the gate number $N\leq 5$ and four sequences for each gate number $5<N\leq 20$. 

\begin{figure*}[tbp]
\centering
\includegraphics[width=1\linewidth]{\figpath 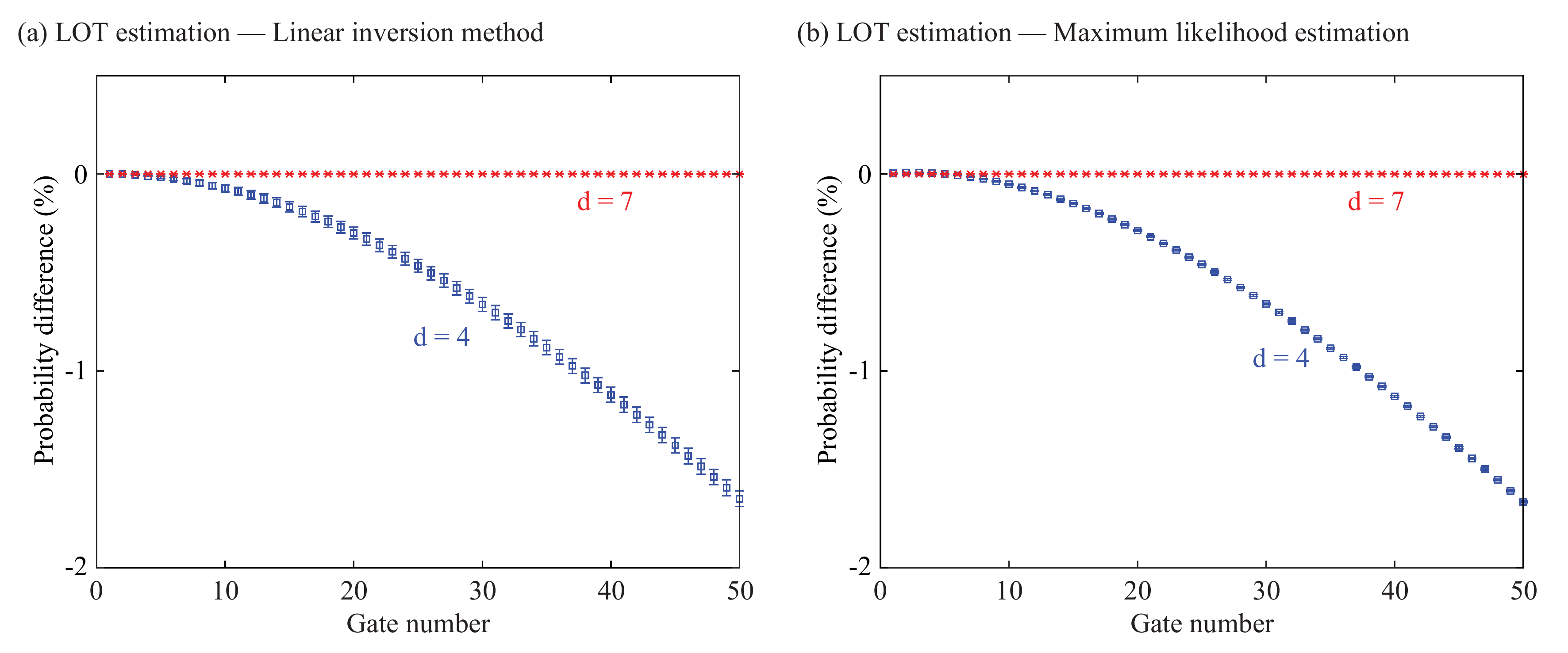}
\caption{
The difference between probabilities in the ideal state $\ket{0}$ after a sequence of randomly chosen gates. (a,b) The probability difference $F_{\rm tom} - F_{\rm act}$, where $F_{\rm act}$ is the probability obtained in the actual quantum computing, and $F_{\rm tom}$ is the probability estimated using linear operator tomography (LOT). The errorbar denotes one standard deviation. 
}
\label{fig:plot_app}
\end{figure*}

\subsection{Maximum likelihood estimation simulation}

Circuits for generating data $\{C_m\}$ used in MLE are the same as circuits used in LIM simulation. 

\subsection{Pauli transfer matrices}

The seven-dimensional Pauli transfer matrices of ideal gates are 
\begin{eqnarray}
M_{\mathcal{O}(H)}^{\rm ideal} &=& \left(\begin{array}{ccccccc}
1 & 0 & 0 & 0 & 0 & 0 & 0 \\
0 & 0 & 0 & 1 & 0 & 0 & 0 \\
0 & 0 & -1 & 0 & 0 & 0 & 0 \\
0 & 1 & 0 & 0 & 0 & 0 & 0 \\
0 & 0 & 0 & 0 & 0 & 0 & 1 \\
0 & 0 & 0 & 0 & 0 & -1 & 0 \\
0 & 0 & 0 & 0 & 1 & 0 & 0 \\
\end{array}\right)
\end{eqnarray}
and 
\begin{eqnarray}
M_{\mathcal{O}(S)}^{\rm ideal} &=& \left(\begin{array}{ccccccc}
1 & 0 & 0 & 0 & 0 & 0 & 0 \\
0 & 0 & 1 & 0 & 0 & 0 & 0 \\
0 & -1 & 0 & 0 & 0 & 0 & 0 \\
0 & 0 & 0 & 1 & 0 & 0 & 0 \\
0 & 0 & 0 & 0 & 0 & 1 & 0 \\
0 & 0 & 0 & 0 & -1 & 0 & 0 \\
0 & 0 & 0 & 0 & 0 & 0 & 1 \\
\end{array}\right)
\end{eqnarray}
In LIM, we take $\widehat{M}_{\rm in} = \widehat{M}_{\rm out}^{-1}g$, and $\widehat{M}_{\rm out}$ is chosen to minimise the difference between $M_{\mathcal{O}}$ and $M_{\mathcal{O}}^{\rm ideal}$.


\begin{thebibliography}{9}
\bibitem{footnote} We use $\norm{\cdot}$ to denote a vector norm satisfying $\abs{\rbraket{A}{B}} \leq \norm{\rbra{A}} \norm{\rket{B}}$ and the submultiplicative matrix norm induced by the vector norm. We can take $\norm{\rket{\cdot}} = \norm{\cdot}_1$, where $\norm{\cdot}_1$ is the trace norm. Then, $\norm{\mathcal{O}} = 1$ if $\mathcal{O}$ is trace-preserving. $\norm{\cdot}_{\rm max}$ denotes max norm.
\end{thebibliography}
\end{document}